\documentclass{article}
\usepackage[utf8]{inputenc}
\usepackage[margin=0.55in]{geometry}
\usepackage{amsmath,amssymb}
\numberwithin{equation}{section}

\usepackage{graphicx,float,color}
\usepackage{enumerate}
\usepackage{xcolor}
\usepackage{soul}
\usepackage[makeroom]{cancel}
\usepackage{enumitem}
\usepackage{tcolorbox}
\usepackage{fullpage}
\usepackage{times}
\usepackage{fancyhdr,graphicx,amsmath,amssymb}
\usepackage[ruled,vlined]{algorithm2e}
\usepackage{tikz}
\usepackage{mypackage}
\usepackage{amsthm}
\usepackage{tabularx}
\usetikzlibrary{positioning}
\usetikzlibrary{arrows}

\DeclareMathOperator*{\argmin}{arg\,min}
\def\en{\ensuremath}
\newcommand{\bs}[0]{\boldsymbol{\beta}}
\newcommand{\bh}[0]{\hat{\boldsymbol{\beta}}}
\newcommand{\bss}[0]{\boldsymbol{\beta}_{S}}
\newcommand{\bhs}[0]{\hat{\boldsymbol{\beta}}_{S}}
\newcommand{\bsms}[0]{\boldsymbol{\beta}_{-S}}
\newcommand{\bhms}[0]{\hat{\boldsymbol{\beta}}_{-S}}

\newcommand{\ahh}[0]{\hat{A}^{1/2}}
\newcommand{\ah}[0]{{A}^{1/2}}
\newcommand{\rh}[0]{\hat{R}}

\newtheorem{theorem}{Theorem}[section]
\newtheorem{corollary}{Corollary}[theorem]
\newtheorem{lemma}[theorem]{Lemma}
\newtheorem{remark}{Remark}
 \newtheorem*{assumption*}{\assumptionnumber}
 \newcommand{\rbar}[0]{\Bar{R}}
 \def\nn{\en{\nonumber}}
\providecommand{\assumptionnumber}{}
\makeatletter
\newenvironment{assumption}[2]
 {%
  \renewcommand{\assumptionnumber}{Assumption #1 {(\emph{\normalfont{#2}})}}%
  \begin{assumption*}%
  \protected@edef\@currentlabel{#1 {(\emph{#2})}}%
 }
 {%
  \end{assumption*}
 }

\usepackage{pdfpages}
\usepackage{tocbibind}
\usepackage{tocloft}

\usepackage{scrlayer}

\DeclareNewLayer[
    foreground,
    contents={%
      \parbox[b][\layerheight][c]{\layerwidth}
        {\centering }%
    }
  ]{blankpage.fg}
\DeclarePageStyleByLayers{blank}{blankpage.fg}

 \newtheorem*{condition*}{\conditionnumber}
\providecommand{\conditionnumber}{}
\makeatletter
\newenvironment{condition}[1]
 {%
  \renewcommand{\conditionnumber}{Condition #1}%
  \begin{condition*}%
  \protected@edef\@currentlabel{#1}%
 }
 {%
  \end{condition*}
 }
 
\begin{document}
 \title{High-dimensional Longitudinal Inference via a De-sparsified Dantzig-Selector}
 \maketitle

\pagenumbering{arabic}

\begin{center}
Nathan W. Huey\textsuperscript{1,2}
\end{center}

\vspace{1mm}

\noindent 1 \textit{Department of Biostatistics, Harvard T.H. Chan School of Public Health, Boston, MA, USA}\\
\noindent 2 \textit{Current Address: Department of Clinical Research, Massachusetts Eye and Ear Hospital, Boston, MA, USA}

\vspace{5mm}

\begin{center}
{\large \bf Abstract} 
\end{center}

We address the problem of valid statistical inference in high-dimensional generalized linear models with longitudinally clustered data. We propose a de-sparsified estimator, constructed from an initial Dantzig-type regularized fit, that enables asymptotically normal inference for individual regression coefficients without requiring model selection consistency. Theoretical results establish validity for both continuous and binary outcomes, and for continuous data under linear models, we prove that the estimator achieves the semiparametric efficiency bound when the working correlation structure is correctly specified. Extensive simulation studies demonstrate accurate coverage and competitive power relative to existing methods. We further illustrate the approach in a genetic association study of bacterial riboflavin production, highlighting its applicability to complex, high-dimensional longitudinal data.

\newpage

\section{Introduction}
In the last two decades, conducting statistical inference on low-dimensional quantities in otherwise high-dimensional settings has received considerable attention. Indeed, this is in no small part influenced by the ever-increasing availability of these kinds of data from fields ranging from economics to signal processing to genetics (\cite{econhighdim}, \cite{sigprochighdim}, \cite{genehighdim}) -- where the abundance of data as well growing computational power has allowed us to explore nuances of high dimensionality that were not possible to address a few decades earlier. While much progress has been made in this arena of high-dimensional statistics, especially when sparse structures are assumed on underlying models, one important area still to be developed full generality is that of inference with high-dimensional longitudinal data, i.e. data with $p \gg n$ ($p$ being the number of variables in the study and $n$ the sample size) and non-independent correlation structures. In the present work we address this area of research and propose a ``de-sparsifying'' method using a preliminary Dantzig selector  type estimator \citep{candes2007dantzig,bickel2009simultaneous} utilizing generalized estimating equation representation of generalized linear models. We present theory first for the linear model which includes demonstrating $\sqrt{n}$-convergence of the proposed estimator to a limiting normal distribution. We demonstrate that among a reasonable class of asymptotically linear estimators, ours attains the lowest possible variance -- and thereby justifying a notion of efficience \citep{jankova2018semiparametric}. Numerical simulations are provided for both continuous and dichotomous outcomes. We conclude with an application of our method to a relevant dataset. First, we place this work in a context of related, though distinct, literature.

\subsection{Related Work}
The method we propose here fits within the general structure developed in \cite{neykov2018unified}. They note three advantages of their method over existing inferential methods in high dimensions. Firstly, many existing methods depend on the specification of a loss function, such as a negative log-likelihood, which has been used to construct an initial estimator of the quantity of interest that one seeks to operate on to produce a de-biased asymptotically normal estimator. While many estimating equations may correspond to the gradient of a loss function (e.g. score functions for parametric models), this is not a strict requirement. Second, \cite{neykov2018unified}'s method is not constrained only to linear and generalized linear models. %As we will work only with moment restriction based models, this point does not apply to the present work. 
Finally, and most importantly for us, we will show that the framework of \cite{neykov2018unified} can be successfully  adapted to our clustered data setting.
 
\vspace{0.5cm}

Before we proceed further, we present a broad overview of the inferential developments for high dimensional sparse generalized linear models. In the context of high-dimensional sparse generalized linear models
\cite{buhlmann2013statistical}, inspired by \cite{zhang2014confidence} developed a de-biasing procedure by using Lasso regression-based initial estimation. In the context of linear regression, \cite{javanmard2018debiasing} showed that asymptotic normality of the debiased LASSO follows with $s =
o_p(n/(\log{p})^2)$ if the rows of the precision matrix of the covariates have sparsity $s_{\Omega}= o_p(
\sqrt{n}/ \log p).$
In contrast, $s_{\Omega}$ is too large, then asymptotic normality follows only if $s =
o_p(\sqrt{n}/(\log{p}))$.
%Both studies remark that they do not rely on the non-zero coefficients being at least of a certain magnitude, the so-called ``uniform signal-strength condition''. \cite{zhang2014confidence} use multiple testing corrections to be able to simultaneously use multiple LDP intervals. 
%\vspace{0.5cm}
Subsequently, \cite{ning2017general} introduced a decorrelated score approach which operates under a suitable loss-based framework. %The difficulty in applying this method to the generalized estimating equation approach would be identifying the ``gradient'' corresponding just to the nuisance parameters of the model. 
It is unclear, how to extend this without a likelihood framework that arises in high dimensional generalized estimating equations. Moreover, \cite{chernozhukov2015valid} discusses orthogonality conditions needed for valid post-regularization inference in high-dimensional models which deals with general linear functionals of nonparametric regressions and their derivatives. Finally, \cite{xia2022statistical} provide inference theory for a de-biased estimator also based on generalized estimating equations. However, Their estimator differs from ours in several points including the choice of initial sparse estimator, where they choose to use a quasi-likelihood-based estimator. They also employ a different tuning procedure compared to our method. They note that solving for a zero as we do may lead to numerical instability in cases other than linear regression. Finally, they do not prove any efficiency results for their estimator. As of the writing of this paper, they do not provide a publicly available implementation of their method.
%\vspace{0.5cm}
Last but not least, \cite{jankova2018semiparametric} develops lower bounds for high-dimensional semiparametric linear and Gaussian graphical models. We crucially adapt this framework to establish the optimality of our methods in the longitudinal data setting.  %Introducing the concept of strong asymptotic unbiasedness of estimators, they subsequently use two approaches for efficiency bounds of such estimates: Cramer-Rao bounds and extensions of Le Cam's arguments for asymptotically linear estimators. In their treatment of the former, they assume Gaussianity on both the error terms and rows of the design matrix. They do not consider dependent data, remaining completely in the context of de-biased lasso and nodewise regression for the linear model and Gaussian graphical model, respectively. Careful adaptations of these arguments would need to include considering $\Omega$ not as the precision matrix of a generative normal distribution of the covariates, but as the Hessian matrix of the estimating equations we consider in this work (except in specific cases). Of particular importance for the high-dimensional setting, they note that the ``worst-possible direction'' for the parameter must lie in the submodel for their bound to hold, corresponding to the sparsity of the precision matrix. While a lower bound can still be given via a sparse approximation of the precision matrix if this condition does not hold, asymptotic efficiency of the de-biased lasso no longer follows. In their treatment of Le Cam's bounds, they assume the existence of a likelihood. While this does not exist for our models, it may be possible to adapt these arguments to our semiparametric setting.
\subsection{Notation}
$\norm{\cdot}_2$ refers to the $L_2$ norm when applied to vectors and the operator/spectral norm of a matrix. For two matrices $A$ and $B$, $A \otimes B$ denotes the Kronecker product. $\norm{A}_{\max}$ refers to the maximum magnitude of the elements of a matrix $A$. $\norm{x}_1$ refers to the $L_1$-norm of a vector, $x$. $\norm{x}_{\infty}$ refers to the  maximum magnitude of the elements of a vector $x$. $\lambda_{\min}(A)$ and $\lambda_{\max}(A)$ refer to the minimum and maximum eigenvalues of a matrix $A$, respectively. For a set of indices $Q$, $x_Q$ refers to the vector with the elements of the indices not in $Q$ replaced with zeros. $x_{-Q}$ refers to the vector with the elements of the indices in $Q$ replaced with zeros. $x_j$ refers to the $j$th element of a vector $x$. $[A]_{i\cdot}$ refers to the $i$th row of the matrix $A$.
\subsection{Contribution}
In the work that follows, we develop an estimator for the $j$th element of $\boldsymbol{\beta}$, $\boldsymbol{\beta}_j$, in the high-dimensional, longitudinal setting that, under certain mild conditions, is asymptotically normal with an $\sqrt{n}$-convergence rate. We then show that, when the cluster-specific correlation structure is correctly specified, our estimator asymptotically attains an efficiency bound within a class of asymptotically unbiased linear estimators, the definition and associated discussions of which are collected in Section \ref{sec:eff}. To our knowledge, this is the first time that a declassified estimator for clustered data settings has been shown to attain such a bound. For example, our estimator is similar to that of \cite{xia2022statistical}, but they do not include any efficiency results. To construct our estimator, we first identify a set of estimating equations such that the expectation at $\boldsymbol{\beta}$ is zero. Since we are in a high-dimensional setting, directly solving for $\boldsymbol{\beta}$ as in classical M-estimation is an ill-posed problem. Therefore, we proceed by using a Dantzig selector-type initial estimate of $\boldsymbol{\beta}$ which enforces a constraint on the $\ell_1$-norm of the parameter estimate, $\hat{\boldsymbol{\beta}}$, inducing sparsity. While the components of this initial estimate are not asymptotically normal due to the effect of shrinking coefficients towards zero as a result of regularization, we ``de-sparsify'' by considering a particular weighted sum of the $p$-estimating equations, chosen such that the resulting estimate asymptotically attains the semi-parametric efficiency bound. All but a single component of the initial sparse estimate of $\boldsymbol{\beta}$ are substituted into this weighted sum, and our estimate of the remaining component is defined as the root of the resulting univariate function as detailed in later sections.

 \vspace{0.5cm}

 \noindent In summary,
the main contributions of this work are as follows:
\begin{enumerate}
    \item We propose a \emph{de-sparsified generalized estimating equations} (de-sparsified GEE) estimator for high-dimensional longitudinal and clustered data. The method starts from an initial Dantzig-type regularized estimator and removes regularization bias through a projection step, enabling valid inference on individual coefficients.
    \item We establish theoretical guarantees for both continuous and binary outcomes. In the continuous-outcome case under linear models, we prove that the proposed estimator attains the semiparametric efficiency bound when the working correlation structure is correctly specified.
    \item We provide a consistent variance estimator for the de-sparsified GEE, allowing for asymptotically valid confidence intervals and hypothesis tests without requiring model selection consistency.
    \item We conduct extensive simulation studies for continuous and binary data, demonstrating accurate coverage, competitive power, and robustness to misspecification of the working correlation structure.
    \item We illustrate the method with an application to a genetic association study of bacterial riboflavin production, showing the method's applicability to complex, high-dimensional longitudinal data.
\end{enumerate}

% ===== drop this where you want the table to appear =====
\section*{Notation}
\begingroup
\setlength{\tabcolsep}{6pt}
\renewcommand{\arraystretch}{1.15}
\begin{tabularx}{\textwidth}{lcc}
\toprule
\textbf{Symbol} & \textbf{Meaning} & \textbf{Dimension / Notes} \\
\midrule
$n$ & Number of clusters & Integer \\
$K$ & Observations per cluster (assumed fixed for theory) & Integer \\
$p$ & Number of covariates & Integer; high-dimensional, $p \gg n$ allowed \\
$Y_{ij}$ & Outcome for cluster $i$, time $j$ & Scalar \\
$\mathbf{X}_{ij}$ & Covariate vector for $(i,j)$ & $\mathbb{R}^{p}$ \\
$\mathbf{Y}_i=(Y_{i1},\ldots,Y_{iK})$ & Outcomes in cluster $i$ & $\mathbb{R}^{K}$ \\
$\mathbf{X}_i = (\mathbf{X}_{i1}\!:\!\cdots\!:\!\mathbf{X}_{iK})^{\top}$ & Design in cluster $i$ & $K\times p$ \\
$\boldsymbol{\beta}=(\beta_1,\ldots,\beta_p)^{\top}$ & Regression coefficients & $\mathbb{R}^{p}$ \\
$s$ & Sparsity of $\boldsymbol{\beta}$ (nonzero count) & Integer \\
$\boldsymbol{\varepsilon}_i=(\varepsilon_{i1},\ldots,\varepsilon_{iK})$ & Cluster errors & $\mathbb{R}^{K}$, $\boldsymbol{\varepsilon}_i \perp \mathbf{X}_i$ \\
$V=\mathrm{Var}(\boldsymbol{\varepsilon}_i)$ & True within-cluster covariance & $K\times K$ \\
$\hat V_n$ & Working covariance estimate & $K\times K$ \\
$\tilde V$ & Probability limit of $\hat V_n$ & $K\times K$ \\
$A=\mathrm{diag}(\hat\sigma_1^2,\ldots,\hat\sigma_K^2)$ & Working variance (GLM) & $K\times K$ \\
$\hat R$ & Working correlation matrix & $K\times K$ \\
$\Sigma_X=\mathbb{E}[\mathbf{X}_i^{\top}\mathbf{X}_i]$ & Covariance of covariates & $p\times p$ \\
$\mathbf{X}_i^{\ell}$ & $\ell$th covariate across $K$ times in cluster $i$ & $\mathbb{R}^{K}$ \\
$e_j$ & Canonical basis vector & $\mathbb{R}^{p}$ with 1 at index $j$ \\
$\hat{\boldsymbol{\beta}}$ & Initial Dantzig selector & solves \eqref{eqn:initialdantzig} \\
$\hat v$ & CLIME/Dantzig column estimator & solves \eqref{eqn:CLIME} \\
$v=[\mathbb{E}\{K^{-1}\mathbf{X}_i^{\top}\tilde V^{-1}\mathbf{X}_i\}]^{-1}_{j\cdot}$ & Target projection vector & $\mathbb{R}^{p}$; sparsity $s_v$ \\
$\hat T_n$ & De-sparsified estimator of $\beta_j$ & defined by \eqref{eqn:theestimate} \\
$\hat{\boldsymbol{\varepsilon}}_i$ & Residual vector for cluster $i$ & $\mathbb{R}^{K}$ \\
$\lambda,\lambda'$ & Tuning parameters & $\lambda \asymp \sqrt{\log p/n}$, etc. \\
$\hat\Delta$ & Variance estimator for $\hat T_n$ & see Theorem \ref{theorem:cltforunknownv} \\
$\hat U_n$ & Studentized statistic & $n^{1/2}\hat\Delta^{-1/2}(\hat T_n-\beta_j)$ \\
\bottomrule
\end{tabularx}
\endgroup
\section{Set-Up}

\subsection{Data and Model}
Assume we have $n$ iid clusters of data, $(Y_{ij},\mathbf{X}_{ij})$ for $i=1,\ldots,n$, $j=1,\ldots,K_i\equiv K$ with $Y_{ij}\in \mathbb{R}$ and $\mathbf{X}_{ij}\in \mathbb{R}^p$ generated as $Y_{ij}=\boldsymbol{{\beta}}^T\mathbf{X}_{ij}+\varepsilon_{ij}$ such that $\boldsymbol{\varepsilon}_i:=(\varepsilon_{i1},\ldots,\varepsilon_{iK})$ satisfies $
\boldsymbol{\varepsilon}_i\perp \mathbf{X}_i:=(\mathbf{X}_{i1}:\cdots:\mathbf{X}_{iK})^T$ and we let $V_K:=\mathrm{Var}(\varepsilon_{i1},\ldots,\varepsilon_{iK})$ to be the homoscedastic joint variance-covariance matrix of the idiosyncratic errors for cluster $i$. Going forward we shall collect the observations across clusters in matrices denoted by $\mathbb{Y}:=(\mathbf{Y}_1:\cdots,\mathbf{Y}_n)^T$ with $\mathbf{Y}_i=(Y_{i1},\ldots,Y_{iK})$ and similarly $\mathbb{X}=(\mathbf{X}_1:\cdots,\mathbf{X}_n)^T$.  % Assume we have $n$ iid clusters of data, $(\mathbf{X}_i, \mathbf{Y}_{i})$, that follow a probability law $\mathbb{P}_{\boldsymbol{\beta}}$, each with $K_i$ observations on $p$ covariates and a univariate continuous or binary outcome. For ease of presentation, we will consider the case where $K_i = K$ for all $i = 1,..n$. For continuous data, we assume a linear model for the expectation of the outcome, $\mathbf{Y}_{i}$, i.e. $\mathbf{Y}_{i} = \mathbf{X}_i\boldsymbol{\beta}+ \boldsymbol{\varepsilon}_i$, $\mathbb{E}[\boldsymbol{\varepsilon}_i|\mathbf{X}_i] = 0$. For binary data, we adopt a logistic model, i.e. $P(\mathbf{Y}_{i} = 1) = \exp(\mathbf{X}_i\bs)(1+\exp(\mathbf{X}_i\bs))^{-1}$.
$\Sigma_X$ denotes the true $p \times p$ covariance matrix associated with $\mathbf{X}_{ij}$, constant across all $i$ and $j$. $V$ is the $K \times K$ covariance matrix associated with $\boldsymbol{\varepsilon}_i$. $s$ is the number of non-zero elements of $\boldsymbol{{\beta}}$. $\mathbf{X}_{i}^\ell$ is the $K \times 1$ vector containing the $\ell$th covariate of cluster $i$ across the $K$ observations. In this context, the goal of this paper is to produce semi-parametric efficient inference on the coordinates of $\boldsymbol{\beta}=(\boldsymbol{\beta}_1,\ldots,\boldsymbol{\beta}_p)^T$.
% \begin{remark}[\bf Working Correlation Matrix Structure]
%     \textnormal{Our estimate of $V$, $\hat{V}_n$, has the form $A^{1/2}\hat{R}A^{1/2}$, where $A = diag(\hat{\sigma}^2_1,...,\hat{\sigma}^2_K)$, $\hat{\sigma}^2_i$ is a moment-based estimator of the variance of the $i$th observation, and $\hat{R}$ is the so-called ``working-correlation'' matrix ( discussed further in Section \ref{section:gee}). These estimates are calculated from an independent sample of $n$ clusters. Regularity conditions on this working correlation matrix estimate will be discussed in the Section \ref{section:assumptions}. }
% \end{remark}

\subsection{Method of Inference}

\subsubsection{Standard Inference Procedures and Generalized Estimating Equations}\label{section:gee}

The use of generalized estimating equations (GEE) as introduced in \cite{zeger1988models} has become extremely popular due in no small measure to its wide applicability and robustness to model misspecification. One must only specify the first two moments of the data-generating distribution of the outcome and under mild regularity conditions, the estimators retain consistency, asymptotic normality, and unbiasedness even when the second moments (specified through the working-correlation matrix) are incorrectly specified. The estimators are defined as the roots to the following equations:
\begin{equation}
    U(\boldsymbol{\beta}) = \sum_{i=1}^n\frac{\partial\mu_i}{\partial\boldsymbol{\beta}}A^{-1/2}\hat{R}^{-1}A^{-1/2}(\mathbf{Y}_{i}- \mu_i(\boldsymbol{\beta}))
\end{equation}
where $A = diag(\hat{\sigma}^2_1,...,\hat{\sigma}^2_K)$, $\hat{\sigma}^2_i$ is an estimator of the variance of the $i$th observation (assumed for convenience here to be constant across clusters, although this assumption can easily be relaxed), and $\hat{R}$ is an estimate of the conditional correlation structure within a cluster. This structure is chosen by the analyst and may be incorrectly specified. In the classical setting where $p$ and $K$ are considered fixed and $n$ is allowed to grow, theoretical results follow via well-known results from M-estimation. As previously mentioned, however, this body of theory is no longer applicable when $p \gg n$, necessitating the development of a modern approach to inference for these models.
% \begin{remark}[\bf Generalization to Exponential Family]
%     \textnormal{We choose to focus on continuous and binary data in the current work due the popularity of these data types, but note that as in \cite{xia2022statistical}, our framework can be generalized to subgaussian exponential family distributions as well (see Section \ref{section:extension}).}
% \end{remark}
\subsubsection{GEE in High-Dimensions}

We now provide the formulations of our estimator, starting with the initial Danzig-type forms.

\vspace{0.5cm}

Let \begin{equation}\label{eqn:initialdantzig}
    \hat{\boldsymbol{\beta}} = \argmin_{\boldsymbol{\beta} \in \mathbb{R}^p} \norm{\boldsymbol{\beta}}_1 \text{, s.t. }\frac{1}{nK}\norm{\sum_{i=1}^n\mathbf{X}_i^T(Y_{i} - \mathbf{X}_i\boldsymbol{\beta})}_{\infty} \leq \lambda
\end{equation}

\noindent and \begin{equation}\label{eqn:CLIME}
    \hat{v} = \argmin_{v \in \mathbb{R}^p} \norm{v}_1 \text{, s.t. }\frac{1}{nK}\norm{\sum_{i=1}^n\mathbf{X}_i^T\hat{V}_n^{-1}\mathbf{X}_iv - e_j}_{\infty} \leq \lambda'
\end{equation}

\noindent where $e_j \in \mathbb{R}^p$ is the $j$th canonical vector and $\lambda, \lambda' > 0$ are tuning parameters. We denote the limiting value of $\hat{V}_n$ as $\Tilde{V}$. Let $v= [\mathbb{E}[K^{-1}\mathbf{X}_i^T\Tilde{V}^{-1}\mathbf{X}_i]]^{-1}_{j \cdot}$. $s_v$ is the number of non-zero elements of $v$.

\vspace{0.25cm}

\noindent Our estimator, $\hat{T}_n$, for the $j$th component of $\boldsymbol{{\beta}}$ is defined such that:
\begin{equation}\label{eqn:theestimate}
    \hat{v}^T\left(\frac{1}{nK}\sum_{i=1}^n\mathbf{X}_i^T\hat{V}_n^{-1}(Y_{i} - \mathbf{X}_i\hat{\boldsymbol{\beta}}_{\hat{T}_n}(j))\right) = 0
\end{equation}
\noindent where $\hat{\boldsymbol{\beta}}_{\hat{T}_n}(j)$ is $\hat{\boldsymbol{\beta}}$ with the $j$th component replaced by $\hat{T}_n$.
\begin{remark}[\bf Initial Estimators]
    \textnormal{It is important that the initial estimates $\hat{\boldsymbol{\beta}}$ and $\hat{v}$ attain specific $L_1$ error rates in order for $\hat{T}_n$ to be asymptotically normal. This is the central reason we choose Dantzig-type estimators and the reason that the estimating equation used for $\hat{\boldsymbol{\beta}}$ is slightly different than that for $\hat{T}_n$, i.e. this estimate is more convenient to work with in theory development and computationally simpler while still attaining the required $L_1$ error rate. While valid inference can still be obtained without a working covariance matrix, we cannot generally achieve optimal efficiency without one.}
\end{remark}

\subsection{ Main Theoretical Results}\label{section:assumptions}
We first discuss several assumptions upon which the theoretical basis of our method rests. 
\begin{assumption}{1}{A1}\label{ass:sparsity}
%$\max\{s, s_v\}\norm{v}_1\log(p)/\sqrt{n} = o(1)$
$\max\{s, s_v\}\log(p)/\sqrt{n} = o(1)$ 
\end{assumption}

\begin{assumption}{2}{A2}\label{ass:subgauss}
    The errors, $\varepsilon_{ij}$, and rows of the design matrix, $\mathbf{X}_{ij}$, are subgaussian with subgaussian norms $K_{\boldsymbol{\varepsilon}}$ and $K_x$, respectively.
\end{assumption}
\begin{assumption}{3}{A3}\label{ass:eigen}
The eigenvalues of $\Sigma_X = \mathbb{E}[\mathbf{X}_i^T\mathbf{X}_i]$, $V$, and $\Tilde{V}$ are bounded above by a universal finite constant and below away from zero, i.e. there exists $M > 0$ such that:
\begin{equation}
    \{\lambda_{\max}(\Sigma_X), \lambda_{\max}(V), \lambda_{\max}(\Tilde{V}) \} < M \text{ and } \{\lambda_{\min}(\Sigma_X), \lambda_{\min}(V), \lambda_{\min}(\Tilde{V}) \} > M^{-1} 
\end{equation}
\end{assumption}
\begin{assumption}{4}{A4}\label{ass:vstarbound}
$\norm{v}_1 = O(1)$
\end{assumption}
\begin{remark}[Discussion of Assumptions]
    \textnormal{Assumption \ref{ass:sparsity} is a common assumption in high-dimensional de-biasing literature. \cite{jankova2018semiparametric} suggests that this level of sparsity is necessary for asymptotic normality of such estimates. Thus, we require only standard sparsity conditions for our method to achieve the desired inferential results. When $V$ is a standard identity matrix, Assumption \ref{ass:vstarbound} means that the columns of $\Sigma_X^{-1}$ have bounded $L_1$-norms. Under a Gaussian graphical model, this precision matrix encodes conditional dependencies among covariates. In this context, the assumption may be thought of as a bound on the amount of collinearity admissable in the design matrix.}
\end{remark}

\begin{theorem}
    \label{theorem:cltforunknownv}
    Under Assumptions \ref{ass:sparsity}, \ref{ass:subgauss}, \ref{ass:eigen}, \ref{ass:vstarbound}, $\lambda \propto \sqrt{\log{p}/n}$, and $\lambda' \propto \norm{v}_1\sqrt{\log p/n}$, the statistic $\hat{U}_n := n^{1/2}\hat{\Delta}^{-1/2}(\hat{T}_n - \boldsymbol{\beta}_j)$ satisfies $\forall t \in \mathbb{R}$:
    \[\lim_{n \to \infty} |\mathbb{P}(\hat{U}_n \leq t) - \Phi(t)| = 0\] where \[\hat{\Delta} = \hat{v}^{T}\frac{1}{n}\sum_{i=1}^n\left[\frac{1}{K^2} \mathbf{X}_i^T\hat{V}_n^{-1}\left[\frac{1}{n}\sum_{i=1}^n\boldsymbol{\hat{\varepsilon}}_i\boldsymbol{\hat{\varepsilon}}_i^T\right]\hat{V}_n^{-1}\mathbf{X}_i\right]\hat{v}\]
and $\hat{V}_n$ is calculated from an independent sample of size $n$.

\begin{proof}
See proof in the Appendix.
\end{proof}
\end{theorem}
\begin{remark}[Effects of Sample-Splitting and Cross-Fitting]
    \textnormal{To avoid losing power due to the sample-splitting used to calculate $\hat{V}_n$ from an independent sample, we will use the procedure of ``cross-fitting'', i.e. performing the estimation procedure twice, each time swapping the roles of the two samples, and averaging the resulting estimates.}
\end{remark}
% \begin{remark}[Proof Strategy] \textnormal{We first show that the estimator assuming $\Tilde{V}$ is known, denoted as $\Tilde{T}_n$, converges to a normal distribution at a $\sqrt{n}$-rate with variance $\Delta$ (Theorem \ref{theorem:intermediatethm}). We then show that $\hat{T}_n$ converges to $\Tilde{T}_n$ with overwhelming probability (Lemma \ref{lemma:Tn}). Finally, showing that we have a consistent estimator of $\Delta$ (Lemma \ref{lemma:deltaconsistentknownv}.), we argue the same $\sqrt{n}$-convergence for $\hat{T}_n$ via an application of Slutsky's theorem.} 
% \end{remark}
\section{Efficiency in High-Dimensional Models}\label{sec:eff}
Efficiency, i.e. statistical optimality in the sense of achieving minimum variance among a certain class of estimators, is a key metric by which estimators for the same underlying structure are compared. As we follow the strategy of \cite{jankova2018semiparametric}, who in turn modify the efficiency arguments of Le Cam, we first provide a very brief overview of these arguments.

%For example, the classical Cramer-Rao lower bounds provide limits on how small the variance among unbiased estimators of parameters can be. Thus, any unbiased estimator achieving this benchmark can be deemed ``best'' in the sense that among all unbiased estimators, it has the smallest expected mean square error.

\subsection{Le Cam's Efficiency Arguments}
We restrict our attention only to a class of data-generating distributions, $\{P_{\theta}^n\}$ that are ``smooth'' in their indexing parameter, $\theta$ which falls in an open subset $\Theta \subset \mathbb{R}^p$ . This condition can be shown to imply that draws from a local neighborhood of distributions about a specific value $\theta = \theta_0$, $\{P_{\theta_0 + h/\sqrt{n}}^n\}$, for any $h \in \mathbb{R}^p$, can be approximated by draws from a Gaussian distribution, $N(h, I_{\theta_0}^{-1})$ asymptotically, where $I_{\theta_0}^{-1}$ is the inverse Fisher Information. This is the concept of Local Asymptotic Normality (LAN). In the course of showing LAN, Taylor expansions of the likelihood are employed with remainder terms that depends on the direction, $h$, of the perturbation around the specific point $\theta_0$ under consideration. Now consider a sequence of statistics, $T_n$ generated under the local models $\{P_{\theta_0 + h/\sqrt{n}}^n\}$ that converge in distribution 
\begin{equation}
    \sqrt{n}(T_n- \psi(\theta_0 + h/\sqrt{n})) \overset{d}{\to}L_{\theta, h}, \text{ for all $h\in \mathbb{R}^p$}
\end{equation}
When $L_{\theta, h}$ does not depend on $h$, $T_n$ are called ``regular'' estimators. This implies that for a sequence of estimators, $T_n$ for $\psi(\theta)$, generated under the local distributions $\{P_{\theta_0 + h/\sqrt{n}}^n\}$, there exists a random variable $T$ with distribution $N(h, I_{\theta_0}^{-1})$, such that $T - \Dot{\psi}_{\theta}h \sim L_{\theta, h}$ for all $h\in \mathbb{R}^p$. Thus, by finding the ``best'' estimator of $\Dot{\psi}_{\theta}h$ in the sense of smallest variance, we will then have the smallest variance possible for an estimator of $\psi(\theta)$ within the class of estimators under consideration. It follows that the best choice for this estimator is $\Dot{\psi}_{\theta}X$ for a single observation $X \sim N(h, I_{\theta_0}^{-1})$. Since $\Dot{\psi}_{\theta}X - \Dot{\psi}_{\theta}h \sim N(0, \Dot{{\psi}}_{\theta}I^{-1}_{\theta}\Dot{{\psi}}_{\theta}^T)$, we conclude that the asymptotic efficiency bound in this class of estimators is $\Dot{{\psi}}_{\theta}I^{-1}_{\theta}\Dot{{\psi}}_{\theta}^T$.
% Such bounds also exist for semi-parametric models, i.e. those with some unstructured functional part of the data-generation distribution, in the classical fixed-dimensional setting (see for example \cite{van2000asymptotic}). In this setting, the semiparameteric efficiency bound can be conceptualized as the supremum of the Cramer-Rao lower bounds for all possible fully parametric submodels, i.e. parametric models covered by the semiparametric case that contain the true data-generating law. 

 \vspace{0.5cm}

 \noindent As noted in \cite{van2000asymptotic}, Section 8.8, although there is no one definition of ``optimality'', any estimator sequence $\sqrt{n}(T_n- \psi(\theta))$ with an $N(0, \Dot{{\psi}}_{\theta}I^{-1}_{\theta}\Dot{{\psi}}_{\theta}^T)$ limiting distribution enjoys multiple optimality properties, the variance cannot be improved except on a set of parameters with Lebesgue measure zero, and any improvement is only possible with specific loss functions and leads to worse performance for others. Le Cam's arguments use the fact that the dimension $p$ is fixed with respect to $n$. We will now note where this assumption is used and how the arguments are adapted to the high-dimensional setting.

\subsection{Adaptations to High-Dimensional Settings} %Rewrite Jankova Section 10 in terms of my problem
In this section, we will define both a class of data-generating laws and a class of estimators for the $j$th element of the parameter vector $\beta$ through sets of conditions that these laws and estimators must satisfy, respectively. Given that our data-generating law falls in the specified class, we will show that there exists a lower bound on the variance within this class of estimators. We will then argue that a multivariate Gaussian data-generating mechanism lies within this class of laws, the estimator proposed in Equation \ref{eqn:theestimate} belongs to this class of estimators, and that this estimator asymptotically achieves the lower bound for this class, in this sense being an efficient estimator. In other words, under the assumption of Gaussianity of the outcomes, this estimator achieves the lowest possible variance within the specified class of estimators.  

\vspace{0.5cm}

\noindent Given subgaussian covariates $\mathbf{X}_1,..., \mathbf{X}_n$, we observe $\mathbf{Y}_1,...,\mathbf{Y}_n$ iid observations from the data-generating process $P_{\beta_{n,0}}$, $\beta_{n,0} \in \mathcal{B}$, which is an open convex subset of $\mathbb{R}^p$. 

In our setting, we consider the class of multivariate Gaussian distributions, $P_{\beta_{n,0}} = \{N(\mathbf{X}_i\beta_{n,0}, V)\}$. We adapt to a high-dimensional setting by only considering $\beta_{n,0} \in \mathcal{B}(d_n)$, where

\begin{equation}\label{eqn:sparseball}
    \mathcal{B}(s_n) = \{\beta \in \mathcal{B}: \norm{\beta}_0 \leq C_1s_n, \norm{\beta}_2 \leq C_2\}
\end{equation}
for $C_1, C_2 = O(1)$ and $s_n$ is a rate that will be derived later. Assume now that we have an estimator of $\beta_j$,$B_n$, satisfying the following conditions:

\begin{condition}{1}{}\label{cond:sp1}
$B_n$ is asymptotically linear, i.e. there exists an influence function $ l_{\beta_{n,0}}$ and some sequence $\beta_{n,0}$ such that:
\begin{equation}
    B_n - \beta_j = \frac{1}{n}\sum_{i=1}^n l_{\beta_{n,0}}(\mathbf{Y}_i) + o_p(n^{-1/2})
\end{equation}
where the final term is with respect to the sequence $\beta_{n,0}$, $\mathbb{E}_{\beta_{n,0}}[l_{\beta_{n,0}}] = 0$, and $V_{\beta_{n,0}}:= \mathbb{E}_{\beta_{n,0}}[l_{\beta_{n,0}}^2] < \infty$.
\end{condition}

The estimator $T_n$, defined in Equation \ref{eqn:theestimate} is asymptotically linear with influence function under correct specification of the working correlation structure (Lemma \ref{lemma:sp1}):
\begin{equation}
    l_{\beta_{n,0}}(\mathbf{Y}_i) = v^{T}K^{-1}\mathbf{X}_i^T{V}^{-1}(\mathbf{X}_i\beta_{n,0} - \mathbf{Y}_{i})
\end{equation}
which has mean zero and bounded variance, $V_{\beta_n,0}$ (Lemma \ref{lemma:sp2}), thus satisfying Condition \ref{cond:sp1}.

\begin{condition}{2}{}\label{cond:sp2}
    The data-generating law, $P_\beta$, must be dominated by a $\sigma$-finite measure. In the following, we denote the distribution function as $p_\beta$, the log-likelihood as $\ell_{\beta}$, and the score function as $s_{\beta}$. The score function must be twice differentiable and $\norm{\Ddot{s}_\beta}_{\max} < L$ for some universal constant $L$ and for all $\beta \in \mathcal{B}(s_n)$. 
    \end{condition}

    \begin{condition}{3}{}\label{cond:sp3}
  Denote $I_{\beta_{n,0}}:= \mathbb{E}_{\beta_{n,0}}[s_{\beta_{n,0}}s_{\beta_{n,0}}^T]$. Assume $\lambda_{\max}(I_{\beta_{n,0}}) = O(1)$, $\lambda_{\min}^{-1}(I_{\beta_{n,0}}) = O(1)$, and 
  \begin{equation}
      \norm{\frac{1}{n}\sum_{i=1}^n\Dot{s}_{\beta_{n,0}} + I_{\beta_{n,0}}}_{\max} = O_p(\lambda)
  \end{equation}
  for some $\lambda > 0$. Assume $s_n = \max\left\{\sqrt{\frac{n}{\log p}}, n^{1/3}\right\}$.
    \end{condition}
 These two conditions are shown for the setting of subgaussian covariates and Gaussian errors in Lemma \ref{lemma:sp2}. Notably we demonstrate that 
\begin{equation}
    \norm{\frac{1}{n}\sum_{i=1}^n \mathbf{X}_i^T{V}^{-1}\mathbf{X}_i - \mathbb{E}[\mathbf{X}_i^T{V}^{-1}\mathbf{X}_i]}_{\max} = O_p\left(\sqrt{\frac{\log p}{n}}\right)
\end{equation}
which defines the sparsity rate $s_n = \max\left\{\sqrt{\frac{n}{\log p}}, n^{1/3}\right\}$ in Equation \ref{eqn:sparseball}.

\begin{condition}{4}{}\label{cond:sp4}
    Let $f_{\beta_{n,0}}:= l_{\beta_{n,0}}+ h^Ts_{\beta_{n,0}}$. Assume that $\forall \varepsilon > 0$:
    \begin{equation}
        \lim_{n \to \infty} \mathbb{E}_{\beta_{n,0}}[f^2_{\beta_{n,0}}\mathbf{1}_{|f_{\beta_{n,0}}|> \varepsilon\sqrt{n}}] = 0
    \end{equation}
    and $1/V_{\beta_{n,0}} = O(1)$
\end{condition}
In Lemma \ref{lemma:sp5}, we show this Lindeberg condition for our data setting and estimator, which is necessary for asymptotic normality of quantities in the likelihood expansion. 

\begin{condition}{5}{}\label{cond:sp5}
    For every $h \in \mathbb{R}^p$:
    \begin{equation}
        \mathbb{E}_{{\beta_{n,0}}}[l_{\beta_{n,0}}h^Ts_{{\beta_{n,0}}} - h^Te_j] = o(1)
    \end{equation}
\end{condition}
Condition \ref{cond:sp5} ensures that the estimator $B_n$ is asymptotically unbiased. We show this condition for $T_n$ in Lemma \ref{lemma:sp6}.
\begin{condition}{6}{}\label{cond:sp6}
    The inverse of $I_{{\beta_{n,0}}}$ exists and ${\beta_{n,0}} + I^{-1}_{{\beta_{n,0}}}e_j/\sqrt{n} \in B({\beta_{n,0}}, c/\sqrt{n})$, where $B({\beta_{n,0}}, c/\sqrt{n})$ is an $\ell_2$-ball of radius $c/\sqrt{n}$ centered at $\beta_{n,0}$ and confined to $\mathcal{B}(s_n)$ for some constant $c > 0$.
\end{condition}
% When considering the local data-generating laws, $\{P_{\theta_0 + h/\sqrt{n}}^n\}$, the direction of perturbation $h$ is unconstrained. As noted above, in the course of showing that a class of distributions in LAN, a likelihood expansion is used with remainder terms depending on the norm of $h$. In classical models, this quantity is well controlled since $p$ is fixed. In the high-dimensional regime, \cite{jankova2018semiparametric} circumvents this issue by restricting the ``local neighborhood'' to sparse directions, specifically to those parameters

\begin{theorem}{(Corollary 2 from \cite{jankova2018semiparametric})}
    Given Conditions \ref{cond:sp1} - \ref{cond:sp6}, 
    \begin{equation}
        \sqrt{n}(B_n - e_j^T(\beta + h/\sqrt{n})e_j)/V_{\beta_{n,0}}^{1/2} \overset{d}{\to} N(0,1)
    \end{equation}
    where $V_{\beta_{n,0}} \geq e_j^TI^{-1}_{{\beta_{n,0}}}e_j + o_p(1)$. This implies that an asymptotically linear estimator with influence function:
    \begin{equation}
        l_{\beta_{n,0}} = e_j^TI^{-1}_{\beta}s_\beta
    \end{equation}
    attains the efficiency bound for this class of estimators.
\end{theorem}
 \vspace{0.5cm}

We can use this result to show an efficiency result for $T_n$:
\begin{theorem}
    \label{theorem:semiparaeffi}
    Under Assumptions \ref{ass:sparsity}, \ref{ass:subgauss}, \ref{ass:eigen}, \ref{ass:vstarbound}, $\lambda \propto \sqrt{\log{p}/n}$, and $\lambda' \propto \sqrt{\log p/n}$, $\beta_{n,0} + v/\sqrt{n} \in B({\beta_{n,0}}, c/\sqrt{n})$ for a sufficiently large constant $c$ and when $Y_{i}|\mathbf{X}_i \sim \mathcal{N}(\mathbf{X}_i\boldsymbol{{\beta}}, V)$, the statistic $\hat{U}_n$ defined in Theorem \ref{theorem:cltforunknownv} asymptotically attains the smallest possible variance among estimators satisfying Conditions \ref{cond:sp1}, \ref{cond:sp4}, and \ref{cond:sp5} when the working correlation structure is correctly specified.
    \begin{proof}
        See proof in Appendix.
    \end{proof}
  
\end{theorem}
 
As noted in \cite{jankova2018semiparametric}, classical results on efficiency were developed for fixed dimensional settings. These are not directly applicable to settings where $p$ is growing with $n$. We show that the estimator, $T_n$ achieves the lowest possible variance for a Gaussian data-generating process with sub-Gaussian covariates in the linear model among a class of estimators satisfying certain regularity conditions in the spirit of Le Cam. 
  \begin{remark}\label{remark:regularsemi}
        \textnormal{As noted by \cite{jankova2018semiparametric}, for the linear model, the sparsity assumption \ref{ass:sparsity} is sufficient for Theorem \ref{theorem:semiparaeffi}, however, other models may require stricter sparsity assumptions for similar results to hold. }
    \end{remark}

\section{Extension to General Linear Models}\label{section:extension}
Although we've focused thus far on the linear model for continuous data, theoretical results can be extended naturally to a general linear model setting. We now briefly outline this extension.
\subsection{Initial Estimator Definition}
Let \[\hat{\boldsymbol{\beta}} = \argmin_{\boldsymbol{\beta} \in \mathbb{R}^d} \norm{\boldsymbol{\beta}}_1 \text{, s.t. }\frac{1}{nK}\norm{\sum_{i=1}^n\sum_{j=1}^K\mathbf{X}_{ij}\left(Y_{ij} - \mu_{ij}(\mathbf{X}_{ij}^T\boldsymbol{\beta})\right)}_{\infty} \leq \lambda\]

Let \[\hat{v} = \argmin_{v \in \mathbb{R}^p} \norm{v}_1 \text{, s.t. }\frac{1}{nK}\norm{\sum_{i=1}^n\mathbf{X}_i^T\hat{A}^{1/2}\hat{R}^{-1}\hat{A}^{1/2}\mathbf{X}_iv - e_j}_{\infty} \leq \lambda'\]

Let $s_v$ denote the sparsity of $v = [E[K^{-2}\mathbf{X}_i^TA^{1/2}\Bar{R}^{-1}A^{1/2}\mathbf{X}_i]]^{-1}_{\cdot j}$, where $\Bar{R}$ is the limiting value of $\hat{R}$.
We note that we now are using a general link function, $\mu_{ij}(\cdot)$, to connect the linear predictor to the mean of the response. The difference in the second Dantzig equation above stems from the mean-variance relationship observed in the exponential family distributions. As a concrete example, consider the case of binary data where we take 
\begin{align*}
    \mu_{ij}(\mathbf{X}_{ij}^T\boldsymbol{\beta}) &= \frac{\exp(\mathbf{X}_{ij}^T\boldsymbol{\beta})}{1+ \exp(\mathbf{X}_{ij}^T\boldsymbol{\beta})}\\
    \frac{\partial \mu_{ij}}{\partial \boldsymbol{\beta}} &= \mu_{ij}(1-\mu_{ij})\mathbf{X}_{ij} = \sigma^2_{j}\mathbf{X}_{ij}
\end{align*}
When we use the identity link function for continuous data, we don't have this same mean-variance dependence, accounting for the difference in the forms of the equations. However, we do note that the identity link function can be thought of as falling into this general framework when the covariates are all scaled to have variance equal to 1 (analagous to scale parameter equal to 1 in the general linear model setting). 

\begin{theorem}
    \label{theorem:cltforunknownvglm}
    Under Assumptions \ref{ass:sparsity}, \ref{ass:subgauss}, \ref{ass:eigen}, \ref{ass:vstarbound}, $\lambda \propto \sqrt{\log{p}/n}$, and $\lambda' \propto \norm{v}_1\sqrt{\log p/n}$, the statistic $\hat{U}_n := n^{1/2}\hat{\Delta}^{-1/2}(\hat{T}_n - \boldsymbol{\beta}_j)$ satisfies $\forall t \in \mathbb{R}$:
    \[\lim_{n \to \infty} |\mathbb{P}(\hat{U}_n \leq t) - \Phi(t)| = 0\] where \[\hat{\Delta} = \hat{v}^{T}\frac{1}{n}\sum_{i=1}^n\left[\frac{1}{K^2} \mathbf{X}_i^T\hat{A}_i^{1/2}\hat{R}^{-1}\hat{A}_i^{-1/2}\boldsymbol{\hat{\varepsilon}}_i\boldsymbol{\hat{\varepsilon}}_i^T\hat{A}_i^{-1/2}\hat{R}^{-1}\hat{A}_i^{1/2}\mathbf{X}_i\right]\hat{v}\]
and $\hat{V}$ is estimated from an independent sample with size $n$.
\begin{proof}
    This proof will follow the same structure as the proof of Theorem \ref{theorem:cltforunknownv}. The primary differences will be found in some of the steps showing the convergence rates for the the working correlation/covariance structure (see for example, Lemmas \ref{lemma:vhatconvergencetovtilde} and \ref{lemma:Deltanrate}) . These steps will also be similar with the forms for the individual elements of these matrices dictated by the specific mean-variance relationship under consideration. As these are standard technicalities, they are omitted in the present work due to space considerations. Note also that another consequence of the mean-variance relationships is now that we cannot estimate a single covariance matrix across all clusters, reflected in the ``meat'' term of the sandwich variance estimate above.
\end{proof}
\end{theorem}

\section{Tuning Parameter Selection}
There are two stages of the estimation procedure which depend on selecting appropriate values for tuning parameters: 1) calculating the initial sparse estimate, $\hat{\boldsymbol{\beta}}$, and 2) calculating the analog of the column of the precision matrix, $\hat{v}$. While the theoretical results above may serve as a starting point, in practice it is important to have a principled way of selecting these parameters. 

\subsection{Choosing $\lambda$}
We tune the initial sparse estimation procedure via 10-fold cross validation across 50 potential tuning parameter values as supplied by the function \emph{glmnet()} from the R package of the same name. The value of $\lambda$ giving the smallest mean squared error is chosen for the analysis.

\subsection{Choosing $\lambda'$}
One of the often-encountered challenges that characterizes de-sparsifying in high-dimensions is that of inverting a high-dimensional matrix (\cite{yuan2010high}). For example, in the linear model setting of our procedure, $v$ is the $j$th column of the high-dimensional matrix, $[\mathbb{E}[K^{-1}\mathbf{X}_i^T\Tilde{V}^{-1}\mathbf{X}_i]]^{-1}$. Since we seek to make inferences on each element of the parameter vector, we need to estimate the entire matrix. In high-dimensions, however, it is impossible to directly solve for the matrix inverse and so other methods are needed. CLIME (\cite{cai2011}) is one such method and is equivalent to solving Equation \ref{eqn:CLIME} for each of the $j = 1,...,p$ columns using the same tuning parameter $\lambda'$. The standard methods for tuning this parameter in the context of CLIME is to use cross-validation across a user-supplied set of potential tuning parameters.  Either a likelihood-based loss function is used or the sum of the squared elements of the trace of the following matrix is used:

\begin{equation}
    \widehat{\Sigma}_{test}\widehat{\Phi}_{train}(\lambda') - I_{p \times p}
\end{equation}
where $\widehat{\Sigma}_{test}$ is the empirical Hessian matrix calculated using the cross-validation test set, $\widehat{\Phi}_{train}(\lambda')$ is the CLIME estimate of the corresponding precision matrix calculated using the training set and tuning value $\lambda'$, and $I_{p \times p}$ is the $p$-dimensional identity matrix. Ideally, the off-diagonal elements of the matrix product will be close to zero, meaning the the sum of the square of the diagonal elements will be close to the Frobenius norm. As we are not assuming a likelihood, we choose this method for the tuning of $\lambda'$

\vspace{0.5cm}

\noindent We investigated the performance of this tuning method via simulation. For each of our $N = 100$ simulations, we generated $2n = 160$ independent draws from a matrix normal distribution of dimension $p = 80$ and $K = 4$ and $AR(0.5)$ row and column covariance matrices, where $\{AR(0.5)\}_{ij} = 0.5^{|i-j|}$ for the $i, j$th element. This models the dependence structure among $p$ variables and between the $K$ observations within a cluster. Outcomes were generated using the linear model with $\boldsymbol{\beta} = (0,1,1,1,1,0,...,0)^T$ and Gaussian errors with error covariance matrix:
\begin{equation*}
    V = \begin{bmatrix}
    1 & 0.5 & 0.5 & 0.5\\
    0.5 & 1 & 0.5 & 0.5\\
    0.5 & 0.5 & 1 & 0.5\\
    0.5 & 0.5 & 0.5 & 1
    \end{bmatrix}
\end{equation*}
We then used $n$ of the observations to estimate $\hat{V}_n$ and the remaining $n$ observations to estimate the precision matrix via CLIME across 50 equally-spaced tuning parameters from 0.01 to 0.5. Instead of estimating the corresponding Hessian, we chose to replace this with a proxy (as the analytical form is cumbersome and unnecessary for our purposes here) for the true Hessian, ${\Sigma}_{true}$, with 10,000 independent draws from the matrix normal distribution and taking the average of the quantities $\{K^{-1}\mathbf{X}_i^TV^{-1}\mathbf{X}_i\}$, for $i = 1,...,10,000$. We then calculated both the Frobenius and max norm errors:
\begin{equation}\label{eqn:tuningerror}
    \norm{{\Sigma}_{true}\widehat{\Phi}(\lambda') - I_{p \times p}}_{\{F, \max\}}
\end{equation}

\begin{remark}
    \textnormal{In the high-dimensional literature, it is common to use so-called ``shrinkage estimators'' (\cite{ledoit2004well}, \cite{daniels2001shrinkage}) when calculating covariance matrices to improve the stability and efficiency of these estimates.  This has been found in recent work (see Chapter 2) to yield lower overall errors similar to Equation \ref{eqn:tuningerror}. Therefore, we also tried using shrinkage estimators in these simulations.}
\end{remark}
\begin{figure}[H]
    \centering
    \includegraphics[scale = 0.3]{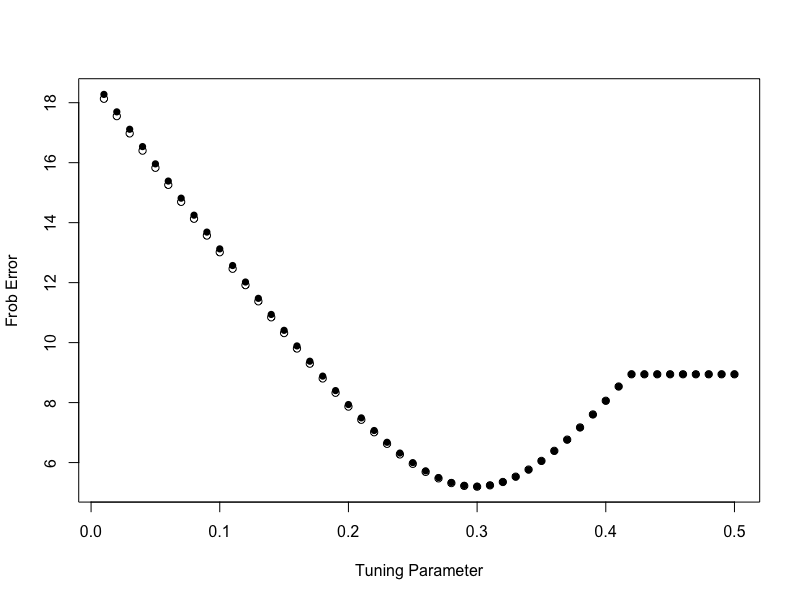}\includegraphics[scale = 0.3]{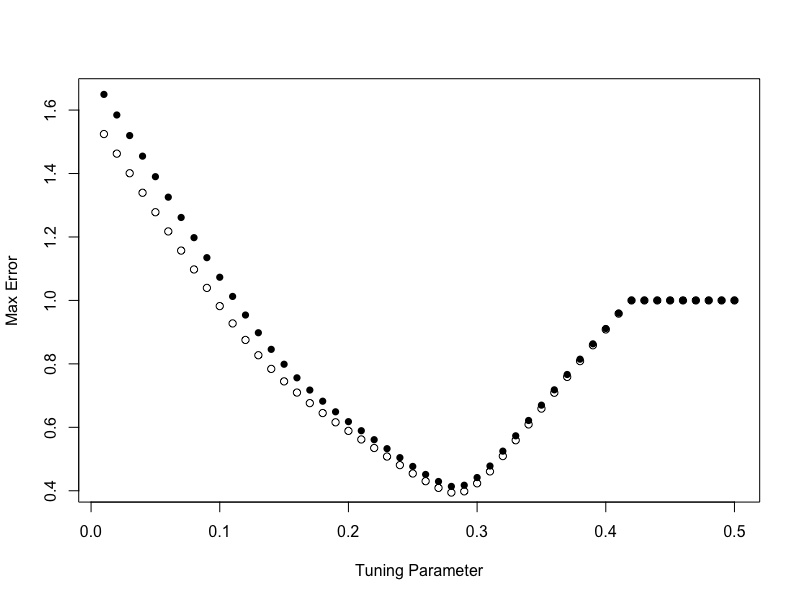}
    \caption{Left: Frobenius Error; Right: Max Error; 10,000 draws were made to estimate the true Hessian; Solid: empirical covariance used; Hollow: Shrinkage estimator used}
    \label{fig:tuningclime}
\end{figure}

From Figure \ref{fig:tuningclime}, we can see that the optimal tuning parameter in this setting appears to be around 0.3 and that using shrinkage estimates does not appear to greatly affect the error. Within each simulation, we used the cross-validation procedure detailed above across the same possible tuning parameter values. All the selected parameters fell between 0.12 and 0.16, indicating that the cross-validation procedure may lead to conservatively sparse solutions. While future work may seek to determine the effect of this on the downstream analysis, we choose to move forward with the cross-validation approach as it is likely better to have conservatively sparse solutions than artificially non-sparse solutions.
\section{Numerical Experiments}

We assessed the performance of our estimator via a number of numerical simulations with both continuous and binary data. Covariates for the $i$th individual were drawn from a matrix normal distribution with $AR(0.5)$ row and column covariance matrices. We compare the proposed method with that of \cite{xia2022statistical}. Since these authors chose not to use sample splitting, we also do not use sample splitting to make the results more comparable. The linear model for the continuous outcome is:
\[\mathbf{Y}_{i} = \mathbf{X}_i\boldsymbol{\beta} + \boldsymbol{\varepsilon}_i\]
where $\boldsymbol{\varepsilon}_i$ is drawn from a multivariate normal distribution with mean 0 and an autoregressive covariance matrix structure for both row and column matrices with parameter $\rho = 0.5$. This models dependence between covariates as well as within clusters. The error matrix, $V$, for continuous data is autoregressive with parameter equal to 0.3. Since the variance parameter is set as uniformly 1, this is also equal to the correlation matrix, $R$. The true $\beta$ vector was $(0,0.5,1,...1,0,...0)^T$ where the number of coefficients equal to 1 was either 3 or 10. The number of clusters was fixed at $K = 4$. To generate correlated clusters of binary responses, we used the function rbin() from the R package SimCorMultRes with simulated latent variables with dependence parameter $\sqrt{0.1}$. We refer the interested reader to the package documentation. 

\vspace{0.5cm}

Initial $\hat{\boldsymbol{\beta}}$ estimates were tuned using 5-fold cross-validation across 10 values of $\lambda$, chosen by the function glmnet() in the R package of the same name. CLIME estimates were tuning using 10-fold cross-validation across 5 equally-spaced values of $\lambda'$ from 0.01 to 0.5. De-sparsified estimates and the estimates of \cite{xia2022statistical} were obtained. Additionally, we employed the method of GEE variance correction in \cite{kauermann2001note} adapted to high-dimensional settings.

\section{Simulation Results}

Due to computational limitations, we omitted the HDIGEE method of \cite{xia2022statistical} for the exchangeable structure in the binary data case for $p = 100$. When $p = 500$, we generally omitted HDIGEE due to longer run times. Particularly for binary data, many simulations encountered computational difficulties and did not complete. We summarize the results of those Monte Carlo simulations that ran to completion in Tables \ref{tab:biascovlen1} - \ref{tab:biascovlen4}. Note that while the AR-1 working correlation structure is correctly specified in the continuous case, this is not necessarily true for binary data, i.e. both working correlation structures tested may be incorrectly specified. In our simulations, the estimates for the intercept term (set equal to zero in all simulations) displayed erratic behavior for our method. We thus remove them from the analysis of the simulation results. 

\vspace{0.5cm}

\noindent Bias of the proposed estimator is comparable or lower than HDIGEE in all settings where both estimators were calculated (Tables \ref{tab:biascovlen1} and \ref{tab:biascovlen3}). In $p = 100$ settings, the variance adjustment method does appear to be effective, increasing the variance and the coverage closer to the desired $95\%$ level, in some settings leading to slightly conservative intervals for the true zeroes (Table \ref{tab:biascovlen1}). In $p = 500$ settings, however, the adjustment doesn't appear to be effective, either leaving the variance largely unchanged or even slightly decreasing it (Tables \ref{tab:biascovlen2} and \ref{tab:biascovlen4}). The resulting coverage values are largely similar or even slightly smaller than the unadjusted intervals. 

\vspace{0.5cm}

\noindent In nearly all settings, coverage values tend to be higher when sparsity is higher, but the bias does not appear to change appreciably when increasing the dimension of the problem, holding other factors constant. Particularly for binary data, bias is larger and coverage is smaller for non-zero values. For $p = 500$ in binary settings, coverage values for non-zero elements ranges from $80-85\%$ while it is around $93-94\%$ for zero-values (Table \ref{tab:biascovlen4}). In the settings where HDIGEE was run, the coverage values were higher than the unadjusted intervals of our method and lower or comparable to the adjusted intervals of our method, with perhaps the exception of the binary data cases (Tables \ref{tab:biascovlen3} and \ref{tab:biascovlen4}). From these simulations, it appears that correct specification of the working correlation matrix in the continuous case leads to slight decreases in confidence interval length with comparable coverage compared to the misspecified exchangeable working correlation settings (Tables \ref{tab:biascovlen1} and \ref{tab:biascovlen2}).

\begin{table}[H]
\centering
\begin{tabular}{lrrrlrrr}
          & \multicolumn{3}{c}{Non-Zero}                                           &           & \multicolumn{3}{c}{Zeroes (no intercept)}                                           \\
          AR-1 (correct) & & &&&&& \\
          & \multicolumn{1}{l}{Bias} & \multicolumn{1}{l}{Coverage} & {Length} &           & \multicolumn{1}{l}{Bias} & \multicolumn{1}{l}{Coverage} & {Length} \\
          & \multicolumn{3}{c}{n = 100, p = 100, s = 3}                                           &           & \multicolumn{3}{c}{n = 100, p = 100, s = 3}                                           \\ \hline
dbgee     & 0.057                             & 0.85                         & 0.203                              &     & 0.047                             & 0.918                        & 0.206                              \\
dbgee.adj & 0.057                             & 0.914                        & 0.242                              &  & 0.047                             & 0.96                         & 0.245                              \\
hdigee    & 0.06                              & 0.893                        & 0.242                              &    & 0.053                             & 0.935                        & 0.249                              \\
          & \multicolumn{3}{c}{n = 100, p = 100, s = 10}                                          &           & \multicolumn{3}{c}{n = 100, p = 100, s = 10}                                          \\ \hline
dbgee     & 0.06                              & 0.809                        & 0.198                              &     & 0.046                             & 0.915                        & 0.199                              \\
dbgee.adj & 0.06                              & 0.875                        & 0.236                              &  & 0.046                             & 0.958                        & 0.237                              \\
hdigee    & 0.06                              & 0.88                         & 0.243                              &    & 0.053                             & 0.928                        & 0.244                              \\
& \multicolumn{1}{l}{}     & \multicolumn{1}{l}{}    & \multicolumn{1}{l}{}       &           & \multicolumn{1}{l}{}     & \multicolumn{1}{l}{}    & \multicolumn{1}{l}{}       \\
Exchangeable & & &&&&& \\
          & \multicolumn{3}{c}{n = 100, p = 100, s = 3}                                          &           & \multicolumn{3}{c}{n = 100, p = 100, s = 3}                                          \\ \hline
dbgee     & 0.058                             & 0.84                         & 0.207                              &     & 0.048                             & 0.916                        & 0.21                               \\
dbgee.adj & 0.058                             & 0.919                        & 0.251                              &  & 0.048                             & 0.964                        & 0.255                              \\
hdigee    & 0.061                             & 0.881                        & 0.245                              &     & 0.053                             & 0.934                        & 0.251                              \\
          & \multicolumn{3}{c}{n = 100, p = 100, s = 10}                                         &           & \multicolumn{3}{c}{n = 100, p = 100, s = 10}                                         \\\hline
dbgee     & 0.061                             & 0.806                        & 0.201                              &     & 0.046                             & 0.914                        & 0.202                              \\
dbgee.adj & 0.061                             & 0.879                        & 0.244                              &  & 0.046                             & 0.961                        & 0.245                              \\
hdigee    & 0.061                             & 0.868                        & 0.245                              &    & 0.053                             & 0.928                        & 0.246                             
\end{tabular}
\caption{Absolute bias, coverage, and $95\%$ CI length for continuous response. ``dbgee'' and ``dbgee.adj'' refer to our proposed method and our method with variance adjustment, respectively. ``hdigee'' is the method of \cite{xia2022statistical}. All values are averages over $N = 250$ Monte Carlo simulations. The left column is averaged over the truly non-zero elements and the right column is averaged over all the true zeroes, excluding the intercept term.}
\label{tab:biascovlen1}
\end{table}

\begin{table}[H]
\centering
\begin{tabular}{lrrrlrrr}
 & \multicolumn{3}{c}{Non-Zero}                                           &           & \multicolumn{3}{c}{Zeroes (no intercept)}                                           \\
 AR-1 (correct) & & &&&&& \\
          & Bias & Coverage & Length &           & \multicolumn{1}{l}{Bias} & \multicolumn{1}{l}{Coverage} & {Length} \\
     N = 249     & \multicolumn{3}{c}{n = 100, p = 500, s = 3}                      &           & \multicolumn{3}{c}{n = 100, p = 500, s = 3}                      \\ \hline
dbgee     & 0.055                    & 0.889                   & 0.219                      &     & 0.047                    & 0.938                   & 0.222                      \\
dbgee.adj & 0.055                    & 0.891                   & 0.22                       &  & 0.047                    & 0.939                   & 0.221                      \\
   N = 250         & \multicolumn{3}{c}{n = 100, p = 500, s = 10}                     &           & \multicolumn{3}{c}{n = 100, p = 500, s = 10}                     \\ \hline
dbgee     & 0.059                    & 0.853                   & 0.214                      &     & 0.046                    & 0.939                   & 0.215                      \\
dbgee.adj & 0.059                    & 0.846                   & 0.214                      & & 0.046                    & 0.939                   & 0.215                      \\
& \multicolumn{1}{l}{}     & \multicolumn{1}{l}{}    & \multicolumn{1}{l}{}       &           & \multicolumn{1}{l}{}     & \multicolumn{1}{l}{}    & \multicolumn{1}{l}{}       \\
 Exchangeable         & \multicolumn{1}{l}{}     & \multicolumn{1}{l}{}    & \multicolumn{1}{l}{}       &           & \multicolumn{1}{l}{}     & \multicolumn{1}{l}{}    & \multicolumn{1}{l}{}       \\
   N = 249         & \multicolumn{3}{c}{n = 100, p = 500, s = 3}                        &           & \multicolumn{3}{c}{n = 100, p = 500, s = 3}                        \\\hline
dbgee     & 0.056                    & 0.89                    & 0.221                      &      & 0.048                    & 0.939                   & 0.224                      \\
dbgee.adj & 0.056                    & 0.889                   & 0.223                      &  & 0.048                    & 0.939                   & 0.224                      \\
         
  N = 249          & \multicolumn{3}{c}{n = 100, p = 500, s = 10}                       &           & \multicolumn{3}{c}{n = 100, p = 500, s = 10}                       \\\hline
dbgee     & 0.06                     & 0.851                   & 0.216                      &     & 0.046                    & 0.939                   & 0.217                      \\
dbgee.adj & 0.06                     & 0.846                   & 0.216                      & & 0.046                    & 0.939                   & 0.216                      \\                  
\end{tabular}
\caption{Absolute bias, coverage, and $95\%$ CI length for continuous response. ``dbgee'' and ``dbgee.adj'' refer to our proposed method and our method with variance adjustment, respectively. Values are averages over $N$ Monte Carlo simulations as noted in the table. The left column is averaged over the truly non-zero elements and the right column is averaged over all the true zeroes, excluding the intercept term.}
\label{tab:biascovlen2}
\end{table}

\begin{table}[H]
\centering
\begin{tabular}{lrrrlrrr}
 & \multicolumn{3}{c}{Non-Zero}                                           &           & \multicolumn{3}{c}{Zeroes (no intercept)}                                           \\
 AR-1 & & &&&&& \\
          & \multicolumn{1}{l}{Bias} & \multicolumn{1}{l}{Coverage} & {Length} &           & \multicolumn{1}{l}{Bias} & \multicolumn{1}{l}{Coverage} & {Length} \\
    N = 237      & \multicolumn{3}{c}{n = 100, p = 100, s = 3}                               &           & \multicolumn{3}{c}{n = 100, p = 100, s = 3}                             \\ \hline
dbgee     & 0.128                    & 0.825                   & 0.442                      &     & 0.096                    & 0.914                   & 0.418                      \\
dbgee.adj & 0.128                    & 0.868                   & 0.482                      & & 0.096                    & 0.941                   & 0.458                      \\
hdigee    & 0.147                    & 0.845                   & 0.535                      &    & 0.121                    & 0.912                   & 0.532                      \\
   N = 86       & \multicolumn{3}{c}{n = 100, p = 100, s = 10}                              &           & \multicolumn{3}{c}{n = 100, p = 100, s = 10}                              \\\hline
dbgee     & 0.146                    & 0.786                   & 0.473                      &     & 0.107                    & 0.914                   & 0.463                      \\
dbgee.adj & 0.146                    & 0.849                   & 0.55                       &  & 0.107                    & 0.952                   & 0.537                      \\
hdigee    & 0.199                    & 0.763                   & 0.584                      &     & 0.139                    & 0.894                   & 0.579                      \\
& \multicolumn{1}{l}{}     & \multicolumn{1}{l}{}    & \multicolumn{1}{l}{}       &           & \multicolumn{1}{l}{}     & \multicolumn{1}{l}{}    & \multicolumn{1}{l}{}       \\
Exchangeable & & &&&&& \\
    N = 121      & \multicolumn{3}{c}{n = 100, p = 100, s = 3}                           &           & \multicolumn{3}{c}{n = 100, p = 100, s = 3}                           \\\hline
dbgee     & 0.123                    & 0.804                   & 0.411                      &      & 0.089                    & 0.918                   & 0.389                      \\
dbgee.adj & 0.123                    & 0.849                   & 0.451                      &  & 0.089                    & 0.941                   & 0.429                      \\
    N = 215      & \multicolumn{3}{c}{n = 100, p = 100, s = 10}                          &           & \multicolumn{3}{c}{n = 100, p = 100, s = 10}                          \\\hline
dbgee     & 0.133                    & 0.811                   & 0.442                      &     & 0.1                      & 0.914                   & 0.434                      \\
dbgee.adj & 0.133                    & 0.865                   & 0.503                      &  & 0.1                      & 0.948                   & 0.494                      \\                   
\end{tabular}
\caption{Absolute bias, coverage, and $95\%$ CI length for binary response. ``dbgee'' and ``dbgee.adj'' refer to our proposed method and our method with variance adjustment, respectively. ``hdigee'' is the method of \cite{xia2022statistical}. Values are averages over $N$ Monte Carlo simulations as noted in the table. The left column is averaged over the truly non-zero elements and the right column is averaged over all the true zeroes, excluding the intercept term.}
\label{tab:biascovlen3}
\end{table}

\begin{table}[H]
\centering
\begin{tabular}{lrrrlrrr}
 & \multicolumn{3}{c}{Non-Zero}                                           &           & \multicolumn{3}{c}{Zeroes (no intercept)}                                           \\
 AR-1 & & &&&&& \\
          & \multicolumn{1}{l}{Bias} & \multicolumn{1}{l}{Coverage} & {Length} &           & \multicolumn{1}{l}{Bias} & \multicolumn{1}{l}{Coverage} & {Length} \\
  N = 238        & \multicolumn{3}{c}{n = 100, p = 500, s = 3}                               &           & \multicolumn{3}{c}{n = 100, p = 500, s = 3}                               \\ \hline
dbgee     & 0.13                     & 0.854                   & 0.467                      &     & 0.096                    & 0.937                   & 0.446                      \\
dbgee.adj & 0.13                     & 0.854                   & 0.467                      &  & 0.096                    & 0.937                   & 0.446                      \\

  N = 226        & \multicolumn{3}{c}{n = 100, p = 500, s = 10}                              &           & \multicolumn{3}{c}{n = 100, p = 500, s = 10}                              \\\hline
dbgee     & 0.144                    & 0.818                   & 0.485                      &     & 0.103                    & 0.934                   & 0.476                      \\
dbgee.adj & 0.144                    & 0.818                   & 0.485                      &  & 0.103                    & 0.934                   & 0.476                      \\
          & \multicolumn{1}{l}{}     & \multicolumn{1}{l}{}    & \multicolumn{1}{l}{}       &           & \multicolumn{1}{l}{}     & \multicolumn{1}{l}{}    & \multicolumn{1}{l}{}       \\
          Exchangeable & & &&&&& \\
  N = 230        & \multicolumn{3}{c}{n = 100, p = 500, s = 3}                           &           & \multicolumn{3}{c}{n = 100, p = 500, s = 3}                           \\\hline
dbgee     & 0.127                    & 0.836                   & 0.43                       &      & 0.088                    & 0.937                   & 0.409                      \\
dbgee.adj & 0.127                    & 0.836                   & 0.43                       &  & 0.088                    & 0.937                   & 0.409                      \\
   
  N = 200        & \multicolumn{3}{c}{n = 100, p = 500, s = 10}                          &           & \multicolumn{3}{c}{n = 100, p = 500, s = 10}                          \\\hline
dbgee     & 0.14                     & 0.79                    & 0.448                      &      & 0.095                    & 0.932                   & 0.439                      \\
dbgee.adj & 0.14                     & 0.79                    & 0.448                      & & 0.095                    & 0.932                   & 0.439                      \\                 
\end{tabular}
\caption{Absolute bias, coverage, and $95\%$ CI length for binary response. ``dbgee'' and ``dbgee.adj'' refer to our proposed method and our method with variance adjustment, respectively. Values are averages over $N$ Monte Carlo simulations as noted in the table. The left column is averaged over the truly non-zero elements and the right column is averaged over all the true zeroes, excluding the intercept term.}
\label{tab:biascovlen4}
\end{table}
\section{Data Application}
We illustrate the use of our method with the \emph{Bacillus subtilis} riboflavin (vitamin B$_2$) data set made publicly available by \cite{buhlmann2014high}. There are $n = 28$ clusters corresponding to distinct strains of the bacterium with two to six observations per cluster for a total of $111$ measurements of the logarithm of riboflavin production. A total of $p = 4088$ gene expression measurements are available per observation. After filtering out those genes whose expression had a coefficient of variation less than $0.1$ (\cite{xia2022statistical}), there remain $p = 267$ genes.

\vspace{0.5cm}

\noindent We chose to avoid sample splitting for the analysis of this data set for several reasons. Firstly, as in our simulation study, we wish to preserve comparability of the results of our method with those of \cite{xia2022statistical}, who did not use sample splitting. Further, the authors of that work did not find that sample splitting made an appreciable difference to their results. Due to the similarities between our method and theirs, we suspect that sample splitting is unnecessary and may lead to potential unwanted consequences of dividing the sample and thereby increasing the aspect ratio. We use an autoregressive working correlation structure for the analysis. We omit the variance correction for the analysis since the simulation studies indicate that this correction may do more harm than good in a model of such high-dimensions.

%We chose to complete several versions of this analysis: 1) Filtering and using AR working correlation 2) Filtering and using independent WC 3) Not filtering using AR working correlation
\begin{figure}[H]
    \centering
    \includegraphics[scale = 0.15]{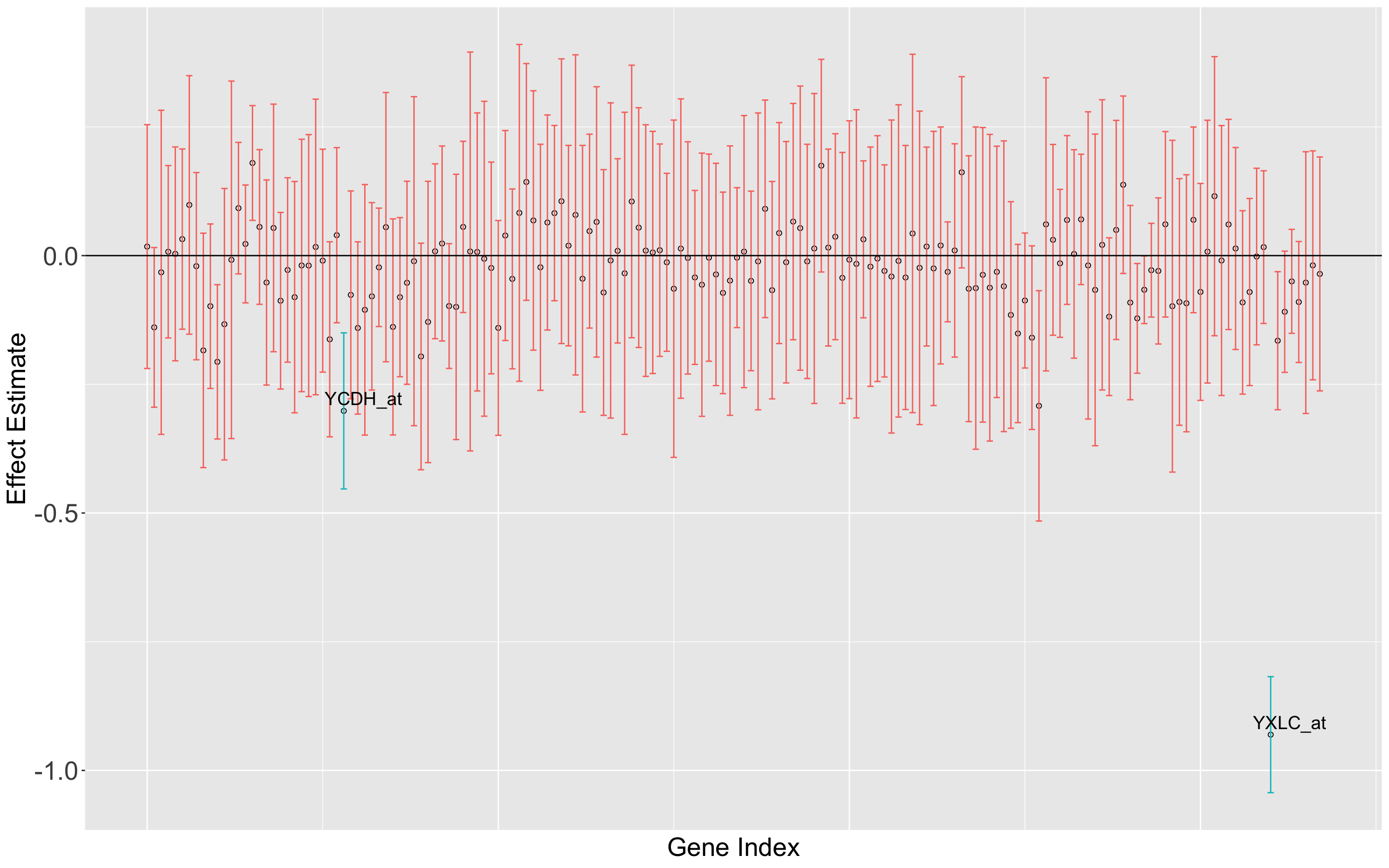}
    \caption{{\bf Confidence Intervals for 167 of 267 genes using an AR working correlation.} The genes with significant Benjamini-Hochberg multiple-testing corrected p-value at a level of 0.05 are highlighted with gene labels next to their confidence intervals.}
    \label{fig:my_labelv7}
\end{figure}

\noindent Using an AR working correlation structure, the proposed method found two genes, \emph{YCDH} and \emph{YXLC}, significant at a 0.05 level (adj.p = $1.3 \times 10^{-2}$ and $1.69 \times 10^{56}$, respectively) using the method of Benjamini and Hochberg to adjust the resulting p-values (Figure \ref{fig:my_labelv7}). The gene \emph{YCDH} codes for a flavin mononucleotide reductase in other bacterial species (\cite{chaston2014metagenome}), potentially implying a related function in \emph{B. subtilis}. The gene \emph{YXLC} is a component of the \emph{sigY} operon in \emph{B. subtilis}, which is involved in making sigma proteins important for the initiation of transcription. This operon can be stimulated by nitrogen starvation (\cite{tojo2003organization}). As riboflavin is a nitrogen containing molecule, this may in part explain a connection between riboflavin production and the \emph{sigY} operon. Other genes involved in this operon are consistently found in other analyses of this data set (\cite{xia2022statistical} and \cite{buhlmann2014high}). \cite{xia2022statistical} performed a similar analysis using their method and the quadratic inference function method of \cite{fang2020test}, finding 2 and 13 significant genes, respectively. We believe that our results complement theirs, providing further evidence for the association of the \emph{sigY} operon with riboflavin production in \emph{B. subtilis}.

\section{Discussion}
In this work, we've introduced a method by which high-dimensional longitudinal data can be effectively analyzed via de-sparsifying of initial Dantzig-selector type regression estimators. Under mild regularity conditions, our estimator has been shown to be asymptotically normally distributed with $\sqrt{n}$-convergence. Additionally, when the working covariance matrix is correctly specified, our estimator asymptotically achieves a semi-parameteric efficiency bound as defined in \cite{jankova2018semiparametric}. We believe that this work represents an important addition to the literature both on generalized estimating equations and high-dimensional inference. For purposes of exposition, our focus was the linear model and we discussed how the theory extends naturally to a generalized linear model. Numerical studies show that our method achieves coverage rates and confidence interval lengths comparable to other alternatives in the literature. Finally, we've shown that our estimator achieves a semiparametric efficiency bound. In this way, among all estimators that converge asymptotically locally around the true parameter value to a Gaussian distribution, our estimator attains the smallest possible variance. 

\vspace{0.5cm}

\noindent Further research directions include investigating the data setting where the cluster size is growing with $n$ as well. There are key points in the theory depending on a fixed cluster size $K$. It will be instructive to determine the rate requirements needed to maintain valid inference. Future extensions should also include performing inference on linear combinations of parameters, i.e. $\xi^T \boldsymbol{\beta}$ for $\xi \in \mathbb{R}^p$ as well as other functions of the parameter vector, e.g. $\norm{\boldsymbol{\beta}}_2$.
%Following the proof in \cite{xia2022statistical}, 

\section{Acknowledgments}
The authors would like to thank Dr. Lu Xia for generously assisting with the implementation of their method.
%\bibliographystyle{natbib}
%\bibliography{second_paper}
%\appendix
%\setcounter{section}{0}
%\renewcommand{\thesection}{\Alph{section}}

\bibliography{references}
\bibliographystyle{apalike}

\appendix
\section{Proofs}

\subsection{Proof of Theorem \ref{theorem:cltforunknownv}}
\begin{theorem}
    Under Assumptions \ref{ass:sparsity}, \ref{ass:subgauss}, \ref{ass:eigen}, \ref{ass:vstarbound}, $\lambda \propto \sqrt{\log{p}/n}$, and $\lambda' \propto \norm{v}_1\sqrt{\log p/n}$, the statistic $\hat{U}_n := n^{1/2}\hat{\Delta}^{-1/2}(\hat{T}_n - \boldsymbol{\beta}_j)$ satisfies $\forall t \in \mathbb{R}$:
    \[\lim_{n \to \infty} |\mathbb{P}(\hat{U}_n \leq t) - \Phi(t)| = 0\] where \[\hat{\Delta} = \hat{v}^{T}\frac{1}{n}\sum_{i=1}^n\left[\frac{1}{K} \mathbf{X}_i^T\hat{V}_n^{-1}\left[\frac{1}{n}\sum_{i=1}^n\boldsymbol{\hat{\varepsilon}}_i\boldsymbol{\hat{\varepsilon}}_i^T\right]\hat{V}_n^{-1}\mathbf{X}_i\right]\hat{v}\]
and $\hat{V}_n$ is calculated from an independent sample of size $n$.

\begin{proof}

% \begin{align*}
%     \mathbb{P}\left(\hat{U}_n \leq t\right) &= \mathbb{P}\left(\hat{U}_n \leq t|\norm{\hat{V}_n - \Tilde{V}}_2 = O\left(\sqrt{\frac{s\log(p)}{n}}\right)\right)\mathbb{P}\left(\norm{\hat{V}_n - \Tilde{V}}_2 = O\left(\sqrt{\frac{s\log(p)}{n}}\right)\right) \\
%     &+ 
%     \mathbb{P}\left(\hat{U}_n \leq t|\norm{\hat{V}_n - \Tilde{V}}_2 \neq O\left(\sqrt{\frac{s\log(p)}{n}}\right)\right)\mathbb{P}\left(\norm{\hat{V}_n - \Tilde{V}}_2 \neq O\left(\sqrt{\frac{s\log(p)}{n}}\right)\right)\\
%     &\overset{n \to \infty}{=} \mathbb{P}\left(\hat{U}_n \leq t| \hat{V}_n: \norm{\hat{V}_n - \Tilde{V}}_2 = O\left(\sqrt{\frac{s\log(p)}{n}}\right)\right)
% \end{align*}
% We state that as the sample size and dimension grow, it becomes overwhelmingly likely that we draw a sequence of data points such that the operator norm of $\hat{V}_n$ converges to that of $\Tilde{V}$ at the specified rate. 

Let $\Tilde{T}_n$ be the estimator constructed using $\Tilde{V}$ instead of the estimated $\hat{V}_n$. Using the fact that
\begin{align*}
   n^{1/2}(\hat{T}_n - \boldsymbol{\beta}_j) &= n^{1/2}(\Tilde{T}_n - \boldsymbol{\beta}_j) + n^{1/2}(\hat{T}_n - \Tilde{T}_n )\\
   &= n^{1/2}(\Tilde{T}_n - \boldsymbol{\beta}_j) + o_p(1)
\end{align*}
by Theorem \ref{theorem:intermediatethm}, we argue that $n^{1/2}(\Tilde{T}_n - \boldsymbol{\beta}_j)$ converges to a normal distribution with variance $\Delta$. From Lemma \ref{lemma:Tn} we have that $n^{1/2}(\hat{T}_n - \Tilde{T}_n ) = o_p(1)$. After showing that $\hat{\Delta}$ is a consistent estimator for $\Delta$ (Lemma \ref{lemma:deltaconsistentknownv}.), our desired result then follows in a standard way from the continuous mapping theorem and Slutsky's Theorem.
\end{proof}
\end{theorem}

\begin{theorem}
    \label{theorem:intermediatethm}
     Under Assumptions \ref{ass:sparsity}, \ref{ass:subgauss}, \ref{ass:eigen}, \ref{ass:vstarbound}, $\lambda \propto \sqrt{\log{p}/n}$, and $\lambda' \propto \sqrt{\log p/n}$, 
    \begin{equation}
        n^{1/2}(\Tilde{T}_n - \boldsymbol{\beta}_j) \overset{d}{\to} N(0, \Delta)
    \end{equation}
    where \[{\Delta} = {v}^{T}\mathbb{E}\left[K^{-1} \mathbf{X}_i^T\Tilde{V}^{-1}V\Tilde{V}^{-1}\mathbf{X}_i\right]{v}\]
Note that while $\hat{\boldsymbol{\beta}}$ is the same estimator as when we estimate the covariance matrix, our estimated projection vector will now be different. We will denote this estimated vector as $\Tilde{v}$.
\begin{proof}
This proof follows by verifying the following six conditions from \cite{neykov2018unified}: 

\begin{enumerate}
    \item \label{cond:conc} (Concentration Rates) For all $T$ in an $\boldsymbol{\varepsilon}$-ball around $\boldsymbol{\beta}_j$, $\mathcal{B}_{\boldsymbol{\varepsilon}}(\boldsymbol{\beta}_j)$,:
    \begin{align}
        \norm{\frac{1}{nK}\sum_{i=1}^n \mathbf{X}_i^T\Tilde{V}^{-1}(\mathbf{X}_i\bs_{T}-\mathbf{Y}_{i})- \mathbb{E}\left[\frac{1}{K}\mathbf{X}_i^T\Tilde{V}^{-1}(\mathbf{X}_i\bs_{T}-\mathbf{Y}_{i})\right]}_{\infty} &= O_p(\sqrt{\log p/n}) = o_p(1) \label{rate:r1}\\
        \norm{v^T\frac{1}{nK}\sum_{i=1}^n \mathbf{X}_i^T\Tilde{V}^{-1}(\mathbf{X}_i\bs_{T}-\mathbf{Y}_{i})- v^T\mathbb{E}\left[\frac{1}{K}\mathbf{X}_i^T\Tilde{V}^{-1}(\mathbf{X}_i\bs_{T}-\mathbf{Y}_{i})\right]}_{\infty} &=\\  O_p(\norm{v}_1\sqrt{\log p/n}) &= o_p(1) \nn \label{rate:r2}\\
        \norm{\Tilde{v}^T\frac{1}{nK}\sum_{i = 1}^n \mathbf{X}_i^T{\Tilde{V}}^{-1}\mathbf{X}_i - e_j}_{\infty} &= \\ O_p(\norm{v}_1\sqrt{\log p/n}) &= o_p(1) \nn \label{rate:r3}\\
        \sup_{T \in \mathcal{B}_{\boldsymbol{\varepsilon}}(\boldsymbol{\beta}_j)}\norm{v^T\left[\mathbb{E}\left[K^{-1}\mathbf{X}_i^T{\Tilde{V}}^{-1}\mathbf{X}_i\right]\right]_{-1}}_{\infty} &< \infty \\
        \sup_{T \in \mathcal{B}_{\boldsymbol{\varepsilon}}(\boldsymbol{\beta}_j)}\norm{\mathbb{E}\left[K^{-1}\mathbf{X}_i^T\Tilde{V}^{-1}(\mathbf{X}_i\bs_{T} - \mathbf{Y}_{i})\right]}_{\infty} &< \infty
    \end{align}
\item \label{cond:l1rates}($L_1$-rates of Initial Estimators)
\begin{align}
    \norm{\bh -\bs}_1 &= O_p\left(s\sqrt{\frac{\log(p)}{n}}\right) \label{rate:r4}\\
    \norm{\Tilde{v} - v}_1 &= O_p\left(\norm{v}_1s_v\sqrt{\frac{\log(p)}{n}}\right) \label{rate:r5}
\end{align}
\item \label{cond:consist} (Consistency of $\Tilde{T}_n$):
\begin{align}
    \Tilde{T}_n \overset{P}{\to}\boldsymbol{\beta}_j
\end{align}

\item \label{cond:clt} (Central Limit Theorem) Assuming $\Tilde{V}_{ii} \geq C_{\boldsymbol{\varepsilon}} > 0$ for $i = 1...K$, $\mathbf{X}_{ij}\perp \boldsymbol{\varepsilon}_i$, and $\lambda_{\min}(\mathbb{E}[\mathbf{X}_i^T\Tilde{V}^{-1}\mathbf{X}_i]) = \xi > 0$ for a fixed constant $\xi$. Assume also that \[\max(s_v, s)\norm{v}_1\log p/\sqrt{n} = o_p(1)\] Then 

\begin{equation}
n^{1/2}\Delta^{-1/2}\left[v^T\sum_{i = 1}^n\mathbf{X}_i^T\Tilde{V}^{-1}\boldsymbol{\varepsilon}_i\right] \overset{d}{\to} N(0,1)
\end{equation}
\item \label{cond:boundedjack}(Bounded Jacobian) There exists a constant $\gamma > 0 $, such that
\begin{equation}
    \left| v^T\frac{\partial}{\partial T} \left[\frac{1}{nK}\sum_{i = 1}^n \mathbf{X}_i^T{\Tilde{V}}^{-1}\mathbf{X}_i\right]_{\cdot 1}\right| < \psi(\mathbb{X}, \mathbf{Y})
\end{equation}
for any $v$ and $\boldsymbol{\beta}$ satisfying $\norm{v - v}_1 < \gamma$ and $\norm{\hat{\boldsymbol{\beta}} - \boldsymbol{\beta}}_1 < \gamma$, where $\psi: \mathbb{R}^{n\times p} \to \mathbb{R}$ is an integrable function $\mathbb{E}[\psi(\mathbb{X}, \mathbf{Y})] < \infty$ under the true data-generating distribution.

\item (Scaling) \label{cond:scaling} Let $r_1$, $r_3$, $r_4$, and $r_5$ be the rates associated with Equations \ref{rate:r1}, \ref{rate:r3}, \ref{rate:r4}, and \ref{rate:r5}, respectively. It follows that:

\begin{equation}
    \sqrt{n}(r_4 \times r_3 + r_5\times r_1) = o(1)
\end{equation}
\end{enumerate}

Condition \ref{cond:conc} is shown in Lemma \ref{lemma:assumption1a}, Corollary \ref{corollary:assumption1b}, and Remark \ref{remark:assumption1ce}. Condition \ref{cond:l1rates} is shown in Theorem \ref{theorem:l1rates} and Corollary \ref{corollary:vl1rate}. Condition \ref{cond:consist} is shown in Lemma \ref{lemma:consistencyconstantv}. Condition \ref{cond:clt} is shown in Theorem \ref{lemma:cltconstantv}. Condition \ref{cond:boundedjack} holds for any linear estimating equation and we note that this condition would need to be checked for general linear models. Condition \ref{cond:scaling} is shown in Lemma \ref{lemma:scaling}. 
\end{proof}
\end{theorem}

\begin{lemma}\label{lemma:Tn}
  $\sqrt{n}(\Tilde{T}_n-\hat{T}_n ) = o_p(1)$
\begin{proof}
$\hat{T}_n$ is the root of $f(T):= \frac{1}{nK}\hat{v}^T\sum_{i=1}^n\mathbf{X}_i^T\hat{V}_n^{-1}(\mathbf{X}_i\hat{\boldsymbol{\beta}}_T-\mathbf{Y}_{i})$. Since this is linear in $T$:
\begin{align*}
    0 &= f(\Tilde{T}_n) + f'(\Tilde{T}_n)(\hat{T}_n - \Tilde{T}_n)\\
    \sqrt{n}(\Tilde{T}_n-\hat{T}_n ) &= \sqrt{n}f'(\Tilde{T}_n)^{-1}f(\Tilde{T}_n)\\
    &= \sqrt{n}\left(\frac{1}{nK}\hat{v}^T\sum_{i=1}^n\mathbf{X}_i^T\hat{V}_n^{-1}\mathbf{X}_ie_j\right)^{-1}\left(\frac{1}{nK}\hat{v}^T\sum_{i=1}^n\mathbf{X}_i^T\hat{V}_n^{-1}(\mathbf{X}_i\hat{\boldsymbol{\beta}}_{\Tilde{T}_n}-\mathbf{Y}_{i})\right)
\end{align*}

The first term in the product is $(1 + o_p(1))^{-1}$ by the definition of $\hat{v}$. So we focus on the second term:

\begin{equation}
    \sqrt{n}\left(\frac{1}{nK}\hat{v}^T\sum_{i=1}^n\mathbf{X}_i^T\hat{V}_n^{-1}(\mathbf{X}_i\hat{\boldsymbol{\beta}}_{\Tilde{T}_n}-\mathbf{Y}_{i})\right) 
\end{equation}

\begin{align*}
   \frac{1}{nK}\hat{v}^T\sum_{i=1}^n\mathbf{X}_i^T\hat{V}_n^{-1}(\mathbf{X}_i\hat{\boldsymbol{\beta}}_{\Tilde{T}_n}-\mathbf{Y}_{i})&=  \frac{1}{nK}\hat{v}^T\sum_{i=1}^n\mathbf{X}_i^T\hat{V}_n^{-1}(\mathbf{X}_i\hat{\boldsymbol{\beta}}_{\Tilde{T}_n}-\mathbf{Y}_{i})  - \frac{1}{nK}\Tilde{v}^T\sum_{i=1}^n\mathbf{X}_i^T\Tilde{V}^{-1}(\mathbf{X}_i\hat{\boldsymbol{\beta}}_{\Tilde{T}_n}-\mathbf{Y}_{i}) \\
   &= \frac{1}{nK}\hat{v}^T\sum_{i=1}^n\mathbf{X}_i^T\Delta_n(\mathbf{X}_i\hat{\boldsymbol{\beta}}_{\Tilde{T}_n}-\mathbf{Y}_{i}) -\\
   &\ \frac{1}{nK}(\Tilde{v}- \hat{v})^T\sum_{i=1}^n\mathbf{X}_i^T\Tilde{V}^{-1}(\mathbf{X}_i\hat{\boldsymbol{\beta}}_{\Tilde{T}_n}-\mathbf{Y}_{i}) \\
   &= \underbrace{\frac{1}{nK}(\hat{v}-\Tilde{v})^T\sum_{i=1}^n\mathbf{X}_i^T\Delta_n(\mathbf{X}_i\hat{\boldsymbol{\beta}}_{\Tilde{T}_n}-\mathbf{Y}_{i})}_{A} - \\
   &\underbrace{\frac{1}{nK}(\Tilde{v}- \hat{v})^T\sum_{i=1}^n\mathbf{X}_i^T\Tilde{V}^{-1}(\mathbf{X}_i\hat{\boldsymbol{\beta}}_{\Tilde{T}_n}-\mathbf{Y}_{i})}_{B} + \\
   &\underbrace{\frac{1}{nK}\Tilde{v}^T\sum_{i=1}^n\mathbf{X}_i^T\Delta_n(\mathbf{X}_i\hat{\boldsymbol{\beta}}_{\Tilde{T}_n}-\mathbf{Y}_{i})}_{C}\\
\end{align*}
% If we have upper RE conditions on $\left(n^{-1}\sum_{1=1}^nK^{-1}\mathbf{X}_i^T\hat{V}_n^{-1}\mathbf{X}_i\right)^2$ such that 
% \begin{equation}
%     \sup_{w \in \mathcal{C}} \frac{w^T\left[\sum_{i = 1}^n (nK)^{-1}\mathbf{X}_i^T\hat{V}_n^{-1}\mathbf{X}_i \right]^2w}{w^Tw} < C^2
% \end{equation}
% where $\mathcal{C}:= \{w \in \mathbb{R}^p: \norm{w_S}_1 \geq \norm{w_{S^C}}_1\}$
\begin{align*}
    |A| &\leq \norm{\hat{v}-\Tilde{v}}_1\left[\norm{\frac{1}{nK}\sum_{i=1}^n \mathbf{X}_i^T\Delta_n\mathbf{X}_i(\hat{\boldsymbol{\beta}}_{\Tilde{T}_n} - \boldsymbol{\beta})+ \frac{1}{nK}\sum_{i=1}^n \mathbf{X}_i^T\Delta_n\boldsymbol{\varepsilon}_i}_{\infty}\right]\\
   &\leq o_p(n^{-1/2})\left[ \underbrace{\norm{\frac{1}{nK}\sum_{i=1}^n \mathbf{X}_i^T\Delta_n\mathbf{X}_i(\hat{\boldsymbol{\beta}}_{\Tilde{T}_n} - \boldsymbol{\beta})}_{\infty}}_{A1} + \underbrace{\norm{\frac{1}{nK}\sum_{i=1}^n \mathbf{X}_i^T\Delta_n\boldsymbol{\varepsilon}_i}_{\infty}}_{A2}\right]\\
   (A1) &\leq \norm{\frac{1}{nK}\sum_{i=1}^n \mathbf{X}_i^T\Delta_n\mathbf{X}_i}_{\max}\norm{\hat{\boldsymbol{\beta}}_{\Tilde{T}_n} - \boldsymbol{\beta}}_{1}
\end{align*}
Note that by Lemma \ref{lemma:vtildetovhat} and Assumption \ref{ass:sparsity}, $\norm{\hat{v}-\Tilde{v}}_1 = o_p(n^{-1/2})$. The first term of A1 can be shown to belong to $O_p\left(\frac{\log(p)}{{n}} \cdot \sqrt{\frac{s\log(p)}{n}}\right)$ by Lemma \ref{lemma:Deltanrate}. The second belongs to $O_p\left(s\sqrt{\frac{\log(p)}{n}}\right)$ by Theorem \ref{theorem:l1rates}.

We have that $\norm{\frac{1}{nK}\sum_{i=1}^n\mathbf{X}_i^T\Delta_n\boldsymbol{\varepsilon}_i}_{\infty} = o_p(1)$ via arguments in Remark \ref{remark:gamma}. Therefore:
\begin{align*}
    \sqrt{n}|A| &\leq \sqrt{n} o_p(n^{-1/2}) (o_p(1) + o_p(1)) = o_p(1)
\end{align*}

Similarly, we can show $\sqrt{n}|B| = o_p(1)$ since $\norm{\frac{1}{nK}\sum_{i=1}^n \mathbf{X}_i^T\Tilde{V}^{-1}\mathbf{X}_i(\hat{\boldsymbol{\beta}}_{\Tilde{T}_n} - \boldsymbol{\beta})}_{\infty} = O_p(1)$ from arguments in \cite{zhou2009restricted}. Finally,

\begin{align*}
    C &= \underbrace{\Tilde{v}^T\frac{1}{nK} \sum_{i=1}^n
\mathbf{X}_i^T\Delta_n\mathbf{X}_i(\hat{\boldsymbol{\beta}}_{\Tilde{T}_n} - \boldsymbol{\beta})}_{C1} -  \underbrace{\Tilde{v}^T\frac{1}{nK} \sum_{i=1}^n
\mathbf{X}_i^T\Delta_n \boldsymbol{\varepsilon}_i}_{C2}\\
|C1| &\leq \norm{\boldsymbol{\beta}_{\Tilde{T}_n} - \boldsymbol{\beta}}_1 \norm{\Tilde{v}^T\frac{1}{nK} \sum_{i=1}^n
\mathbf{X}_i^T\Delta_n\mathbf{X}_i}_{\infty}
\end{align*}

The first term in C1 belongs to $O_p\left(s\sqrt{\frac{\log(p)}{n}}\right)$ by \ref{theorem:l1rates}. From arguments from \cite{zhou2009restricted}, we can conclude that the second term in C1 belongs to $o_p(n^{-1/2})$. The fact that C2 belongs to $o_p(n^{-1})$ follows from an application of a Hoeffding-type inequality combined with the convergence assumption for $\Delta_n$. Thus we can conclude that $\sqrt{n}|C| = o_p(1)$, concluding the proof that $\sqrt{n}(\Tilde{T}_n-\hat{T}_n ) = o_p(1)$.
\end{proof}
\end{lemma}

\subsection{Proof of Theorem \ref{theorem:semiparaeffi}}
  \begin{theorem}
    Under Assumptions \ref{ass:sparsity}, \ref{ass:subgauss}, \ref{ass:eigen}, \ref{ass:vstarbound}, $\lambda \propto \sqrt{\log{p}/n}$, and $\lambda' \propto \sqrt{\log p/n}$, $\beta_{n,0} + v/\sqrt{n} \in B({\beta_{n,0}}, c/\sqrt{n})$ for a sufficiently large constant $c$ and when $Y_{i}|\mathbf{X}_i \sim \mathcal{N}(\mathbf{X}_i\bs, V)$, the statistic $\hat{U}_n$ defined in Theorem \ref{theorem:cltforunknownv} asymptotically attains the smallest possible variance among estimators satisfying Conditions \ref{cond:sp1}, \ref{cond:sp4}, and \ref{cond:sp5} when the working correlation structure is correctly specified.
\end{theorem}
    \begin{proof}
        We show the following conditions needed for the result in Corollary 2 of \cite{jankova2018semiparametric}:
        \begin{enumerate}
            \item (Influence Function) Our estimator, $T_n$, is asymptotically linear with influence function 
            \begin{equation}
                \ell_{\bs}(\mathbf{X}_i) = v^TK^{-1}\mathbf{X}_i^T\hat{V}_n^{-1}(\mu_i - \mathbf{Y}_{i})
            \end{equation}

            \item (Bounded Variance) Under Assumption \ref{ass:vstarbound}, it follows that:
            \begin{equation}
                \mathbb{E}[\ell^2_{\bs}(\mathbf{X}_i)] = Var(\ell_{\bs}) < \infty
            \end{equation}
            \item (Score Function) The score function associated with the data-generating process, $s_{\boldsymbol{\beta}}$, is twice-differentiable and bounded above.
            \item (Information Matrix) The information matrix, $I_{\boldsymbol{\beta}} = \mathbb{E}[\mathbf{X}_i^T\Tilde{V}^{-1}\mathbf{X}_i]$ has bounded eigenvalues and 
\begin{align}
    \norm{\frac{1}{n}\sum_{i=1}^n \mathbf{X}_i^T\Tilde{V}^{-1}\mathbf{X}_i - \mathbb{E}[\mathbf{X}_i^T\Tilde{V}^{-1}\mathbf{X}_i]}_{\max} &= o_p(1) 
\end{align}
\item (Technical Condition 1) $\forall \boldsymbol{\varepsilon} >0$
\begin{equation}
    \lim_{n \to \infty} \mathbb{E}_{\bs}[f^2_{\bs}\mathbb{I}_{|f_{\bs}|>\boldsymbol{\varepsilon}\sqrt{n}}] = 0
\end{equation}
where $f_{\bs}(x):= \ell_{\bs}(x) + h^Ts_{\bs}(x)$ for $x \in \mathbb{R}^p$ and $h = \sqrt{n}(\Tilde{\boldsymbol{\beta}} - \bs)$ for $\Tilde{\boldsymbol{\beta}}\in B(\bs, \frac{c}{\sqrt{n}})$. Here $B(\bs, \frac{c}{\sqrt{n}})$ is an $L_2$ ball of length $\frac{c}{\sqrt{n}}$ around $\bs$.
\item  (Technical Condition 2) $\forall h \in \mathbb{R}^p$:
\begin{equation}
    \mathbb{E}_{\bs}[\ell_{\bs}h^Ts_{\bs}] - h^Te_j = o(1)
\end{equation}
\item (Technical Condition 3)
\begin{equation}
    \boldsymbol{\beta}+ I^{-1}_{\bs}e_j/\sqrt{n} \in B(\bs, c/\sqrt{n})
\end{equation}
        \end{enumerate}

Since all necessary conditions have been shown in Lemmas \ref{lemma:sp1} - \ref{lemma:sp7},  by observing that the influence function:
\[\ell_{\boldsymbol{\beta}} = \left[K^{-1}\mathbb{E}_{\bs}[\mathbf{X}_i^TV^{-1}\mathbf{X}_i]\right]^{-1}_{\cdot j} K^{-1}\mathbf{X}_i^TV^{-1}(Y-\mathbf{X}_i\boldsymbol{\beta}) = e_j^TI_{\boldsymbol{\beta}}^{-1}s_{\boldsymbol{\beta}}\]
 we use the results from Section 10 of \cite{jankova2018semiparametric} to conclude that our estimator achieves the semiparametric efficiency bound asymptotically.
\end{proof}

\section{Appendix B: Proofs of Important Lemmas}
All proofs in this section are under Assumptions \ref{ass:sparsity}, \ref{ass:eigen}, \ref{ass:subgauss}, and \ref{ass:vstarbound}.
\subsection{Proofs of Lemmas needed for Theorem \ref{theorem:cltforunknownv}}
We now present central lemmas needed to connect $\Tilde{T}_n$ and $\hat{T}_n$:
\begin{lemma}\label{lemma:vhatconvergencetovtilde}

    \begin{equation}
        \norm{\hat{V}_n^{-1} - \Tilde{V}^{-1}}_2 = O_p\left(s\sqrt{\frac{\log(p)}{n}}\right)
    \end{equation}

    \begin{proof}
        \begin{align*}
            \norm{\hat{V}_n^{-1} - \Tilde{V}^{-1}}_2 &= \norm{\hat{V}_n^{-1}( \Tilde{V}- \hat{V}_n )\Tilde{V}^{-1}}_2\\
            &\leq \norm{\hat{V}_n^{-1}}_2 \norm{\Tilde{V}- \hat{V}_n}_2\norm{\Tilde{V}^{-1}}_2
        \end{align*}
        Focusing on the middle term:
        \begin{align*}
            \norm{\Tilde{V}- \hat{V}_n}_2 & \leq \norm{(\ahh - \ah)(\rh - \rbar)(\ahh - \ah) + (\ahh - \ah)(\rh - \rbar)\ah}_2 + \\
            & \norm{(\ahh - \ah)\rbar\ah + (\ahh - \ah)\rbar(\ahh - \ah) + \ah(\rh - \rbar)(\ahh - \ah)}_2 + \\
            &\norm{\ah (\rh - \rbar)\ah + \ah \rbar (\ahh - \ah) }_2
        \end{align*}
        Each of the terms in this summation is made up of three components. Each term is either a difference between an estimate and its limiting quantity or a fixed matrix with bounded eigenvalues. Since both $\hat{A}$ and $\hat{R}$ are calculated with residuals using $\hat{\boldsymbol{\beta}}$ via moment-based estimators, we consider just the rate of $\norm{\ah - \ahh}_2$. Note that for the $j$th  diagonal element of $A - \hat{A}$:
        \begin{align*}
            \left|\frac{1}{n}\sum_{i=1}^n(\mathbf{X}_{ij}^T\bh - Y_{ij})^2 - \sigma_j^2\right| &= \left|\frac{1}{n}\sum_{i=1}^n\mathbf{X}_{ij}^T(\bh - \bs)(\bh -\bs)^T\mathbf{X}_{ij} - 2\mathbf{X}_{ij}^T(\bh - \bs)\boldsymbol{\varepsilon}_{ij} + \boldsymbol{\varepsilon}_{ij}^{2}- \sigma_j^2\right|\\
            &\leq \left|\frac{1}{n}\sum_{i=1}^n\mathbf{X}_{ij}^T(\bh - \bs)(\bh -\bs)^T\mathbf{X}_{ij}\right| + \left|\frac{1}{n}\sum_{i=1}^n2\mathbf{X}_{ij}^T(\bh - \bs)\boldsymbol{\varepsilon}_{ij}\right| +\\
            &\ \left| \frac{1}{n}\sum_{i=1}^n\boldsymbol{\varepsilon}_{ij}^{2}- \sigma_j^2\right|\\
            &\leq \norm{\bh - \bs}_1^2 \norm{\frac{1}{n}\sum_{i=1}^n\mathbf{X}_{ij}\mathbf{X}_{ij}^T}_{\max} + 2\norm{\bh - \bs}_1\norm{\frac{1}{n}\sum_{i=1}^n\mathbf{X}_{ij}^T\boldsymbol{\varepsilon}_{ij}}_{\infty} + o_p(1)\\
            &\leq O_p\left(s\sqrt{\frac{\log(p)}{n}}\right)\left[O(1) + O_p\left(\sqrt{\frac{\log(p)}{n}}\right)\right] + \\
            &O_p\left(s\sqrt{\frac{\log(p)}{n}}\right) \cdot O_p\left(\sqrt{\frac{\log(p)}{n}}\right) + o_p(1)\\
            &= O_p\left(s\sqrt{\frac{\log(p)}{n}}\right)
        \end{align*}
        The key observation here is that the rate ultimately is determined by the rate of convergence of $\bh$ through the residuals. This kind of proof will hold for elements of $\hat{R}$ as well. The slowest rate above will be determined by those terms with a single matrix difference in operator norm. Since each of the involved matrices is $K \times K$, i.e. has a finite and fixed number of elements, we can easily show that the rate of convergence of the operator norm is also $O_p\left(s\sqrt{\frac{\log(p)}{n}}\right)$. Since $\norm{\hat{V}_n^{-1}}_2 = O(1) + o_p(1)$, the final rate is $ O_p\left(s\sqrt{\frac{\log(p)}{n}}\right)$.
    \end{proof}

\end{lemma}

\begin{lemma}\label{lemma:Deltanrate}
Let $\Delta_n = \hat{V}_n^{-1} - \Tilde{V}^{-1}$.
\begin{equation}
    \norm{\frac{1}{nK}\sum_{i = 1}^n \mathbf{X}_i^T\Delta_n\mathbf{X}_i}_{\max} = O_p\left(\frac{\log(p)}{{n}} \cdot \sqrt{\frac{s\log(p)}{n}}\right)
\end{equation}
\begin{proof}
Let $\mathbf{X}_i^m \in \mathbb{R}^{K\times 1}$ denote the $m$th column of the $K \times p$ design matrix for the $i$th cluster.
\begin{align*}
    &P\left\{ \norm{\frac{1}{nK}\sum_{i = 1}^n \mathbf{X}_i^T\Delta_n\mathbf{X}_i}_{\max} \geq \boldsymbol{\varepsilon} \right\}\\
    &\leq p^2 P\left\{\frac{1}{nK}\sum_{i = 1}^n \mathbf{X}_i^{mT}\Delta_n\mathbf{X}_i^{\ell} \geq \boldsymbol{\varepsilon} \right\}\\
    &= p^2P\left\{ \frac{1}{nK}\sum_{i = 1}^n\frac{1}{4}(\mathbf{X}_i^m + \mathbf{X}_i^{\ell})^T\Delta_n(\mathbf{X}_i^m + \mathbf{X}_i^{\ell}) - \frac{1}{4}(\mathbf{X}_i^m - \mathbf{X}_i^{\ell})^T\Delta_n(\mathbf{X}_i^m - \mathbf{X}_i^{\ell})  \geq \boldsymbol{\varepsilon} \right\} \\
    &\leq p^2P\left\{ \frac{1}{nK}\sum_{i = 1}^n\frac{1}{4}(\mathbf{X}_i^m + \mathbf{X}_i^{\ell})^T\Delta_n(\mathbf{X}_i^m + \mathbf{X}_i^{\ell}) + \frac{1}{4}(\mathbf{X}_i^m - \mathbf{X}_i^{\ell})^T\Delta_n(\mathbf{X}_i^m - \mathbf{X}_i^{\ell})  \geq \boldsymbol{\varepsilon} \right\} \\
    &\leq p^2P\left\{ \frac{1}{nK}\sum_{i = 1}^n\frac{1}{4}(\mathbf{X}_i^m + \mathbf{X}_i^{\ell})^T\Delta_n(\mathbf{X}_i^m + \mathbf{X}_i^{\ell}) \geq \frac{\boldsymbol{\varepsilon}}{2} \cap 
    \frac{1}{4}(\mathbf{X}_i^m - \mathbf{X}_i^{\ell})^T\Delta_n(\mathbf{X}_i^m - \mathbf{X}_i^{\ell})  \geq \frac{\boldsymbol{\varepsilon}}{2} \right\} \\
    &\leq p^2P\left\{ \frac{1}{nK}\sum_{i = 1}^n\frac{1}{4}(\mathbf{X}_i^m + \mathbf{X}_i^{\ell})^T\Delta_n(\mathbf{X}_i^m + \mathbf{X}_i^{\ell}) \geq \frac{\boldsymbol{\varepsilon}}{2}\right\} +\\
    &p^2P\left\{\frac{1}{nK}\sum_{i = 1}^n\frac{1}{4}(\mathbf{X}_i^m - \mathbf{X}_i^{\ell})^T\Delta_n(\mathbf{X}_i^m - \mathbf{X}_i^{\ell})  \geq \frac{\boldsymbol{\varepsilon}}{2} \right\}
\end{align*}
    We bound these probabilities using the Hanson-Wright Inequality (\cite{hansonwright}):
    \begin{align*}
        &\leq p^2\left[2\exp\left\{-c\min\left\{\frac{\boldsymbol{\varepsilon}^2\cdot n^3K}{4K_x^4s\log(p)}, \frac{\boldsymbol{\varepsilon} K \sqrt{n^3}}{K_x^2\sqrt{s\log(p)}}\right\}\right\}\right]
    \end{align*}
    where $c$ is a positive absolute constant, using the fact that the convergence rate of $\norm{\Delta}_2$ is $\sqrt{\frac{s\log(p)}{n}}$ by Lemma \ref{lemma:vhatconvergencetovtilde}. So we have that 
    \begin{align*}
        &P\left\{ \norm{\frac{1}{nK}\sum_{i = 1}^n \mathbf{X}_i^T\Delta_n\mathbf{X}_i}_{\max} \leq \boldsymbol{\varepsilon} \right\} \geq\\
        &1 - 2p^2\exp\left\{-c\min\left\{\frac{\boldsymbol{\varepsilon}^2\cdot n^3K}{4K_x^4s\log(p)}, \frac{\boldsymbol{\varepsilon} K \sqrt{n^3}}{K_x^2\sqrt{s\log(p)}}\right\}\right\}
    \end{align*}
    We can solve for $\boldsymbol{\varepsilon}$ such that this probability converges to 1. Upon analytically solving for $\boldsymbol{\varepsilon}$ when either of the above terms is the minimum, we find that the solution has a dominating term of the order $\frac{\log(p)}{n} \cdot \sqrt{\frac{s\log(p)}{n}}$.
\end{proof}
\end{lemma}
\begin{lemma}\label{lemma:vtildetovhat}
\begin{equation}
    \norm{\Tilde{v} - \hat{v}}_1 = O_p\left(s_v{\frac{\log(p)}{n}} \cdot \sqrt{\frac{s\log(p)}{n}}\right)
\end{equation}
    \begin{proof}
        First we show that $\norm{\hat{v}- v}_1 = O_p\left(s_v{\frac{\log(p)}{n}} \cdot \sqrt{\frac{s\log(p)}{n}}\right)$ and the result follows from an application of the Triangle Inequality. This proof is identical to that of Corollary \ref{corollary:vl1rate}, but the estimate $\hat{V}_n$ is substituted for its limiting value $\Tilde{V}$. The difference in rates follows from the calculation of the following probability:

        \begin{align*}
            \mathbb{P}\left\{\norm{v^T\left[\frac{1}{nK}\sum_{i=1}^n\mathbf{X}_i^T\hat{V}_n^{-1}\mathbf{X}_i - \mathbb{E}[K^{-1}\mathbf{X}_i^T\Tilde{V}^{-1}\mathbf{X}_i]\right]}_{\infty} \geq \lambda'\right\}
        \end{align*}
        By setting $\lambda' = \norm{v}_1 \left(4A_xK_{\Tilde{x}}^2\sqrt{\frac{\log(p)}{n}} + \norm{\frac{1}{nK}\sum_{i=1}^n\mathbf{X}_i^T\Delta_n\mathbf{X}_i}_{\max}\right)$, the remainder of the proof is identical, with the new rate determined by the new $\lambda'$ and is $O_p(s_v\lambda')$. By Lemma \ref{lemma:Deltanrate}, then we can conclude that $\norm{\hat{v}- v}_1 = O_p\left(s_v{\frac{\log(p)}{n}} \cdot \sqrt{\frac{s\log(p)}{n}}\right)$.
    \end{proof}
\end{lemma}

We need the following additional condition on the design matrix to show the required $L_1$ rates of the initial estimators.
\begin{lemma}{(Restricted Eigenvalue (RE) Condition on $\mathbb{X}$)}

 There exists $\kappa > 0 $ such that:
    \begin{equation}\label{recondition}\kappa^2 =\min \frac{\norm{\mathbb{X} \delta}_2^2}{nK\norm{\delta_S}_2^2}  > 0, \forall \delta \in \mathbb{R}^p \neq 0 \text{ s.t. } \norm{\delta_S}_1 \geq \norm{\delta_{-S}}_1
    \end{equation}
    where $|S| \leq \min\{s, s_v\}$. Note that as long as $\lambda_{\min}(\Tilde{V}) >0 $ (Assumption \ref{ass:eigen}), this implies a RE condition on $\Tilde{\mathbb{X}} = (\Tilde{V}^{-1/2} \otimes I_{n})\mathbb{X}$. 
    
    \begin{proof}
        This condition is satisfied with high probability in the case of a subgaussian random matrix (\cite{zhou2009restricted}).
    \end{proof}
\end{lemma}
\subsection{Required Lemmas}
\begin{lemma}
\label{lemma:lemma8ai}
     \[\frac{1}{nK}\norm{\mathbb{X}^T\mathbb{X}(\bh - \bs)}_{\infty} \leq 2\lambda\]
     \end{lemma}
     \begin{proof}
   Since we know the subgaussian nature of $\boldsymbol{\varepsilon}$ and $\mathbb{X}$, we can say that \\ $\frac{1}{nK}\norm{\sum_{i=1}^n\sum_{j=1}^K\mathbf{X}_i(Y_{ij} - \mathbf{X}_i^T\boldsymbol{\beta})}_{\infty}\leq \lambda$ occurs with at least a certain probability, \[\gamma = 1 - Kep^{1-\frac{cA^2\log (p)}{2(1+2C A_x)K_x^2}} - 2Kp^{2-\Bar{c}A_x^2}\] where $A_x$ is an arbitrary constant satisfying $A_x\sqrt{\log{p}/n}\leq 1$ (addressed in Remark \ref{remark:gamma}). We know by definition that \[\frac{1}{nK}\norm{\sum_{i=1}^n\sum_{j=1}^K\mathbf{X}_i(Y_{ij} - \mathbf{X}_i^T\bh)}_{\infty}\leq \lambda\] The result follows by an application of the Triangle Inequality with probability at least $\gamma$.
   \end{proof}
\begin{remark}[Deriving $\gamma$]
\label{remark:gamma}
    \textnormal{If there were only a single observation per individual, a Hoeffding's Inequality is could be thus applied:
    \[P\left(\norm{\frac{1}{n}\mathbb{X}^T\boldsymbol{\varepsilon}}_{\infty} \geq t |\mathbb{X}\right) \leq ep\exp\left(-\frac{cnt^2}{K_{\boldsymbol{\varepsilon}}^2\norm{\Sigma_n}_{max}}\right)\]
    where $\boldsymbol{\varepsilon} = [\boldsymbol{\varepsilon}_1 ... \boldsymbol{\varepsilon}_n]^T$ is a vector of independent, centered subgaussian random variables and $K_{\boldsymbol{\varepsilon}}$ is the subgaussian norm of $\boldsymbol{\varepsilon}$. }
    \end{remark}
    \noindent If instead we consider our clustered data setting:
    
    \[P\left( \frac{1}{nK}\norm{X^T\boldsymbol{\varepsilon}}_{\infty}\geq t\right) = P\left( \frac{1}{nK}\norm{\sum_{j=1}^K \sum_{i = 1}^n \mathbf{X}_{ij}\boldsymbol{\varepsilon}_{ij}}_{\infty}\geq t\right)\]
    the vector $\boldsymbol{\varepsilon}$ is no longer a vector of independent subgaussian random variables (due to intra-cluster dependence). However, we can use the fact that the errors conditional on a specific time point are indeed independent:
    
    \begin{align*}
        P\left( \frac{1}{nK}\norm{\sum_{j=1}^K \sum_{i = 1}^n \mathbf{X}_{ij}\boldsymbol{\varepsilon}_{ij}}_{\infty}\geq t\right) &\leq P\left( \frac{1}{K}\sum_{j=1}^K\norm{ \frac{1}{n}\sum_{i = 1}^n \mathbf{X}_{ij}\boldsymbol{\varepsilon}_{ij}}_{\infty}\geq t\right) \rightarrow \text{Tri. Ineq.} \\
        &\leq P\left( \max_{j \in 1...K}\norm{ \frac{1}{n} \mathbf{X}_{j}^T\boldsymbol{\varepsilon}_{j}}_{\infty}\geq t\right) \rightarrow \text{max $>$ mean}
    \end{align*}
    where $\mathbf{X}_j = [\mathbf{X}_{1j} ... \mathbf{X}_{nj}]^T \in \mathbb{R}^{n \times p}$ and $\boldsymbol{\varepsilon}_j = [\boldsymbol{\varepsilon}_{1j}...\boldsymbol{\varepsilon}_{nj}]^T \in \mathbb{R}^n$.
    \begin{align*}
        &\leq \sum_{j=1}^K P\left(\norm{ \frac{1}{n}\mathbf{X}_{j}^T\boldsymbol{\varepsilon}_{j}}_{\infty} \geq t \right) \rightarrow \text{Union bound}\\
        &\leq \sum_{j=1}^K ep\exp\left(-\frac{cnt^2}{K_{\boldsymbol{\varepsilon}}^2\norm{\Sigma_n^j}_{max}}\right) \rightarrow \text{Hoeffding Inequality}
    \end{align*}
    where $K_{\boldsymbol{\varepsilon}}$ is the subgaussian norm of $\boldsymbol{\varepsilon}_j$ and $\Sigma_n^j = n^{-1}\mathbf{X}_j^T\mathbf{X}_j$. Choose $t = AK_{\boldsymbol{\varepsilon}}\sqrt{\log{p}/n}$.
     \begin{align*}
         &= \sum_{j=1}^Kep\exp\left(-\frac{cA^2\log (p)}{\norm{\Sigma_n^j}_{\max}}\right)
     \end{align*}

    Noting that 
    \begin{align*}
        \norm{\Sigma_n^j}_{\max} - \norm{\Sigma_X}_{\max} &\leq \norm{\Sigma_n^j - \Sigma_X}_{\max}\\
        \norm{\Sigma_n^j}_{\max} &\leq 4A_xK_x^2\sqrt{\log{(p)}/n} + 2K_x^2 
    \end{align*}
    where the last line follows with probability at least $1 - 2p^{2-\Bar{c}A_x^2}$ by Lemma \ref{lemma:lemma5} with $V = I$.
    the above becomes
    \begin{align*}
        &\leq \sum_{j=1}^Kep^{1-\frac{cA^2\log (p)}{2(1+2C A_x)K_x^2}}\\
        &= Kep^{1-\frac{cA^2\log (p)}{2(1+2C A_x)K_x^2}}
    \end{align*}
    where $\sqrt{\log{(p)}/n} \leq C$. Thus the probability of this event occurring is at least $1 - Kep^{1-\frac{cA^2\log (p)}{2(1+2C A_x)K_x^2}} - 2Kp^{2-\Bar{c}A_x^2}$.
\begin{lemma}[Sums of Subgaussian are always Subgaussian]
\label{lemma:subgaussian}
If $X$ and $Y$ are subgaussian random variables (possibly dependent) with subgaussian norm bounds $\sigma_x$ and $\sigma_y$, then $X+Y$ is a subgaussian random variable with subgaussian norm bound $\sigma_x + \sigma_y$.
\begin{proof}
Trivial using definition of subgaussian norm.
 %    Using Holder's Inequality:
 %    \begin{align*}
 %        \mathbb{E}[e^{t(X+Y)}] \leq \left(\mathbb{E}[e^{tXp}]^{1/p}\right)\left(\mathbb{E}[e^{tY(1-p)}]^{1/(1-p)}\right)
 %    \end{align*}
 % and applying subgaussianity of the individual variables, we get a bound that is a function of $p$. Then we find the value of $p$ which gives the tightest bound. 
 %(https://math.stackexchange.com/questions/3866298/sum-of-sub-gaussian-random-variables/38663243866324)
\end{proof}

\end{lemma}
\begin{lemma}[Scalar products of Subgaussian are always Subgaussian]
\label{lemma:scalarsubgaussian}
If $X$ is a subgaussian random variable with subgaussian norm bounds $\sigma_x$ , then $cX$ is a subgaussian random variable with subgaussian norm bound $c\sigma_x$.
\begin{proof}
    \begin{align*}
        \mathbb{E}[e^{t(cX)}] \leq \exp \left\{ \frac{c^2\sigma_x^2t^2}{2}\right\}
    \end{align*}
 and it follows directly from the definition of subgaussianity that the subgaussian norm bound of $cX$ is $c\sigma_x$.
\end{proof}

\end{lemma}
\begin{lemma}
\label{lemma:knownvlemma8aii}

    Let $\Tilde{\mathbf{X}}_i = {\Tilde{V}}^{-1/2}\mathbf{X}_i$ and $\Tilde{\mathbf{X}}_{ij}$ be the $j$th row of $\Tilde{\mathbf{X}}_i$. Note that the elements of $\Tilde{\mathbf{X}}_{ij}$ are still subgaussian random variables (Lemmas \ref{lemma:subgaussian} and \ref{lemma:scalarsubgaussian}). Let the relevant subgaussian norm be $K_{\Tilde{x}}$. Letting $\Sigma := \mathbb{E}[K^{-1}\mathbf{X}^T_i{\Tilde{V}}^{-1}\mathbf{X}_i]$ and $\lambda' = \norm{v}_14A_xK_{\Tilde{x}}^2\sqrt{\log{p}/n}$, we have with probability at least $1-2Kp^{2-\Bar{c}A_x^2}$:
    \begin{equation}
    \label{eqn:vineq}
        \norm{\frac{1}{nK}\sum_{i = 1}^n\Tilde{\mathbf{X}}_i^T\Tilde{\mathbf{X}}_iv - e_1}_{\infty} = \norm{v^T\left[\frac{1}{nK}\sum_{i = 1}^n\Tilde{\mathbf{X}}_i^T\Tilde{\mathbf{X}}_i - \Sigma\right]}_{\infty} \leq \lambda'
    \end{equation}
    \end{lemma}
\begin{proof}
    \begin{align*}
        P\left(\norm{v^T\left[\frac{1}{nK}\sum_{i = 1}^n\Tilde{\mathbf{X}_i}^T\Tilde{\mathbf{X}_i} - \Sigma\right]}_{\infty} \geq t\right) &\leq P\left(\norm{v}_1\norm{\left[\frac{1}{nK}\sum_{i = 1}^n\Tilde{\mathbf{X}}_i^T\Tilde{\mathbf{X}}_i - \Sigma\right]}_{\max} \geq t \right) \\
        &\leq P\left(\norm{v}_1\frac{1}{K}\sum_{j = 1}^K\norm{\left[\frac{1}{n}\sum_{i = 1}^n\Tilde{\mathbf{X}}_{ij}\Tilde{\mathbf{X}}_{ij}^T - \Sigma\right]}_{\max} \geq t \right) \\
        &\leq P\left(\norm{v}_1\max_{j \in 1,...K}\norm{\left[\frac{1}{n}\sum_{i = 1}^n\Tilde{\mathbf{X}}_{ij}\Tilde{\mathbf{X}}_{ij}^T - \Sigma\right]}_{\max} \geq t \right)\\
        &\leq \sum_{j = 1}^KP\left(\norm{v}_1\norm{\left[\frac{1}{n}\sum_{i = 1}^n\Tilde{\mathbf{X}}_{ij}\Tilde{\mathbf{X}}_{ij}^T - \Sigma\right]}_{\max} \geq t \right)\\
    \end{align*}

From Lemma \ref{lemma:lemma5}, with probability at least $1-2p^{2-\Bar{c}KA_x^2}$:

\[\norm{\left[\frac{1}{nK}\sum_{i = 1}^n\Tilde{\mathbf{X}_i}^T\Tilde{\mathbf{X}_i} - \Sigma\right]}_{\max} \leq 4A_xK_{\Tilde{x}}^2\sqrt{\log{p}/n}\]

where $\Bar{c} >0$ is an absolute constant and $A_x$ is an arbitrary constant satisfying $A_x\sqrt{\log{p}/n}\leq 1$.
Thus, setting $\lambda' = \norm{v}_14A_xK_{\Tilde{x}}^2\sqrt{\log{p}/n}$, we can find a lower bound of the probability of Equation \ref{eqn:vineq}, namely $1-2Kp^{2-\Bar{c}A_x^2}$.
\end{proof}
    \begin{lemma}\label{lemma:bs1}
    Recall that $\beta_S$ refers to the vector $\beta$ with the corresponding elements not in the indexing set $S$ set to zero and $\bhms$ refers to the vector $\beta$ with the elements in the indexing set $S$ set to zero. Here $S$ is the set of indices corresponding to the $s$ non-zero elements of $\bs$.
    \label{lemma:lemma8b}
    \[\norm{\bh_S - \bs_S}_1 \geq \norm{\bh_{-S} - \bs_{-S}}_1\]
    \end{lemma}
    \begin{proof} We start with an application of the Triangle Inequality:
    \begin{align*}
        \norm{\bhs}_1 + \norm{\bss - \bhs}_1 &\geq \norm{\bss}_1 \\
        \norm{\bhs}_1 + \norm{\bhs - \bss}_1 &\geq \norm{\bss}_1 \\
        \norm{\bhs}_1 &\geq  \norm{\bss}_1 - \norm{\bhs - \bss}_1 \\
    \end{align*}
    Note that $\norm{\bhs}_1 = \norm{\bh}_1 - \norm{\bhms}_1$ and $\norm{\bhms}_1 = \norm{\bhms - \bsms}$, since $\bsms = 0$.\\
    Therefore,
    \[\norm{\bh}_1 - \norm{\bhms - \bsms }_1 \geq \norm{\bss}_1 - \norm{\bhs - \bss}_1\]
    Since $\norm{\bs}_1 = \norm{\bss}_1$, 
    \[\norm{\bh}_1 \geq \norm{\bs}_1 - \norm{\bhs - \bss}_1 +  \norm{\bhms - \bsms }_1 \]
    Since by definition $\norm{\bh}_1 \leq \norm{\bs}_1$,
    \[\norm{\bhs - \bss}_1 \geq \norm{\bhms - \bsms }_1\] as desired to show.
    \end{proof}
    \begin{lemma}\label{lemma:vhats} Let $S_v$ be the set of indices corresponding to the $s_v$ non-zero elements of $v$.
\[\norm{\Tilde{v}_{S_v} - v_{S_v}}_1 \geq \norm{\Tilde{v}_{-S_v} - v_{-S_v}}_1\]
\begin{proof}
    This proof structure is the same as that of Lemma \ref{lemma:bs1}.
\end{proof}
\end{lemma}
\begin{theorem}[$\ell_1$ Rates of Initial Estimators]
\label{theorem:l1rates}
\[\norm{\bh -\bs}_1 = O_p\left(s\sqrt{\frac{\log(p)}{n}}\right)\]
\end{theorem}
\begin{proof}
Consider 
\begin{align*}
    \frac{1}{nK}\norm{X(\bh - \bs)}_2^2 &= \frac{1}{nK}(\bh - \bs)^TX^TX(\bh-\bs)\\
    &\leq \frac{1}{nK}\norm{\bh - \bs}_1\norm{X^TX(\bh - \bs)}_{\infty} \rightarrow \text{H\"{o}lder's}\\
    &\leq {2\lambda}{\norm{\bh - \bs}_1} \rightarrow \text{Lemma \ref{lemma:lemma8ai}, with probability at least $\gamma$}\\
    &= {2\lambda}\left({\norm{\bhs - \bss}_1 + \norm{\bhms - \bsms}_1}\right)\\
    &\leq {2\lambda}\cdot 2\norm{\bhs - \bss}_1 \rightarrow \text{Lemma \ref{lemma:lemma8b}}\\
    &\leq{4\lambda}\sqrt{s}\norm{\bhs - \bss}_2 \rightarrow \text{relationship between 1-and 2-norms}\\
\end{align*}
From Condition \ref{recondition} and Lemma \ref{lemma:lemma8b}:
\[\kappa^2\norm{\bhs - \bss}_2^2 \leq \frac{1}{nK}{\norm{X(\bh - \bs)}_2^2}\]
So 
\begin{align*}
    \kappa^2\norm{\bhs - \bss}_2^2 &\leq {4\lambda}\sqrt{s}\norm{\bhs - \bss}_2\\
    \norm{\bhs - \bss}_2 &\leq \frac{4\lambda}{\kappa^2}\sqrt{s}
\end{align*}
Note that plugging this result back into the above yields:
\begin{align*}
    \frac{1}{nK}\norm{X(\bh - \bs)}_2^2 \leq{4\lambda}\sqrt{s} \cdot \frac{4\lambda}{\kappa^2}\sqrt{s} &= O_p\left(\frac{s\log(p)}{n}\right) \rightarrow \text{choosing $\lambda$ as below}
\end{align*}
Finally, 
\begin{align*}
    \norm{\bh - \bs}_1 &= \norm{\bhs - \bss}_1 + \norm{\bhms - \bsms}_1\\
    &\leq 2\norm{\bhs - \bss}_1 \rightarrow \text{Lemma \ref{lemma:lemma8b}}\\
    &\leq 2\sqrt{s}\norm{\bhs - \bss}_2 \rightarrow \text{relationship between 1-and 2-norms}\\
    &\leq \frac{8\lambda}{\kappa^2 }s \rightarrow \text{above result}
\end{align*}
When we choose $\lambda = C\sqrt{\frac{\log(p)}{n}}$ for a constant $C$ (related to the probability $\gamma$ - see Remark \ref{remark:gamma}),
\[\norm{\bh - \bs}_1 \leq \frac{8C}{\kappa^2}s\sqrt{\frac{\log(p)}{n}} = O_p\left(s\sqrt{\frac{\log(p)}{n}}\right)\]
as was desired to show.
\end{proof}
In a parallel fashion, we can show the following corollary:
\begin{corollary}
\label{corollary:vl1rate}
    \begin{equation}
    \norm{\Tilde{v} - v}_1 = O_p\left(\norm{v}_1s_v\sqrt{\frac{\log(p)}{n}}\right)
\end{equation}
\end{corollary}

\begin{corollary}
Given the definition of $\Tilde{v}$, we can conclude via the Triangle inequality that with probability at least $1-2Kp^{2-\Bar{c}A_x^2}$:
\begin{equation}
    \frac{1}{nK}\norm{\sum_{i = 1}^n\Tilde{\mathbf{X}_i}^T\Tilde{\mathbf{X}_i} (\Tilde{v}-v)}_{\infty} \leq 2\lambda'
\end{equation}
\end{corollary}

\begin{lemma}[Concentration Inequalities]
\label{lemma:assumption1a}
For all $T$ in an $\boldsymbol{\varepsilon}$-ball around $\boldsymbol{\beta}_j$,:
\[\norm{\frac{1}{nK}\sum_{i=1}^n \mathbf{X}_i^T\Tilde{V}^{-1}(\mathbf{X}_i\bs_{T}-\mathbf{Y}_{i})- \mathbb{E}\left[\frac{1}{K}\mathbf{X}_i^T\Tilde{V}^{-1}(\mathbf{X}_i\bs_{T}-\mathbf{Y}_{i})\right]}_{\infty} = O_p(\sqrt{\log(p)/n}) = o_p(1)\]
\begin{proof}
\begin{align*}
    &\norm{\frac{1}{nK}\sum_{i=1}^n \mathbf{X}_i^T\Tilde{V}^{-1}\mathbf{X}_i\bs_{T}- \mathbf{X}_i^T\Tilde{V}^{-1}(\mathbf{X}_i\boldsymbol{\beta}+ \boldsymbol{\varepsilon}_i) - \left[\frac{1}{K}\mathbb{E}\left[\mathbf{X}_i^T\Tilde{V}^{-1}\mathbf{X}_i\right]\bs_{T} - \frac{1}{K}\mathbb{E}\left[\mathbf{X}_i^T\Tilde{V}^{-1}\mathbf{X}_i\right]\bs\right]}_{\infty} \\
    &= \norm{\left[\frac{1}{n}\sum_{i=1}^n\frac{1}{K}\mathbf{X}_i^T\Tilde{V}^{-1}\mathbf{X}_i\right](\bs_{T}- \bs) - \frac{1}{nK}\sum_{i=1}^n\mathbf{X}_i^T\Tilde{V}^{-1}\boldsymbol{\varepsilon}_i - \frac{1}{K}\mathbb{E}\left[\mathbf{X}_i^T\Tilde{V}^{-1}\mathbf{X}_i\right](\bs_{T}- \bs)}_{\infty}\\
    &{\leq \norm{\left[\frac{1}{n}\sum_{i=1}^n\frac{1}{K}\mathbf{X}_i^T\Tilde{V}^{-1}\mathbf{X}_i\right]_{\cdot 1} - \mathbb{E}\left[\frac{1}{K}\mathbf{X}_i^T\Tilde{V}^{-1}\mathbf{X}_i\right]_{\cdot 1}}_{\infty}\boldsymbol{\varepsilon} + \norm{\frac{1}{nK}\sum_{i=1}^n\mathbf{X}_i^T\Tilde{V}^{-1}\boldsymbol{\varepsilon}_i}_{\infty}}
\end{align*}
The first term in the sum is $O_p(\sqrt{\log(p)/n})$ by Lemma \ref{lemma:knownvlemma8aii}. From Lemma \ref{lemma:lemma8ai} and noting that the entries of $\mathbf{X}_i^T\Tilde{V}^{-1}$ remain subgaussian, the second term is also $O_p(\sqrt{\log(p)/n})$. Thus, the sum is $O_p(\sqrt{\log(p)/n}) = o_p(1)$.
\end{proof}
\end{lemma}
A parallel proof shows 
\begin{corollary}[Assumption 1b]
\label{corollary:assumption1b}
    \[ \norm{v^T\frac{1}{nK}\sum_{i=1}^n \mathbf{X}_i^T{V}^{-1}(\mathbf{X}_i\bs_{T}-\mathbf{Y}_{i})- v^T\mathbb{E}\left[\frac{1}{K}\mathbf{X}_i^T{V}^{-1}(\mathbf{X}_i\bs_{T}-\mathbf{Y}_{i})\right]}_{\infty} = O_p(\norm{v}_1\sqrt{\log(p)/n}) = o_p(1) \]
\end{corollary}

\begin{remark}[Assumption 1c-e] 
\label{remark:assumption1ce}
 (c) Choosing $\lambda' \propto \norm{v}_1\sqrt{\log(p)/n}$, it also follows immediately from the definition of $\Tilde{v}$ that 
\[\norm{\Tilde{v}^T\frac{1}{nK}\sum_{i = 1}^n \mathbf{X}_i^T{\Tilde{V}}^{-1}\mathbf{X}_i - e_1}_{\infty} = O_p(\norm{v}_1\sqrt{\log(p)/n}) = o_p(1)\] 
\noindent (d) Additionally:
\[\sup_{T \in N_{\boldsymbol{\beta}_j}}\norm{v^T\left[\mathbb{E}\left[K^{-1}\mathbf{X}_i^T{\Tilde{V}}^{-1}\mathbf{X}_i\right]\right]_{-1}}_{\infty} = 0 < \infty \]
where $A_{-1}$ denotes a matrix $A$ with the first column removed.

\vspace{0.5cm}

\noindent (e) Finally, 
\begin{align*}
    \sup_{T \in N_{\boldsymbol{\beta}_j}}\norm{\mathbb{E}\left[K^{-1}\mathbf{X}_i^T\Tilde{V}^{-1}(\mathbf{X}_i\bs_{T} - \mathbf{Y}_{i})\right]}_{\infty} &= \sup_{T \in N_{\bs_j}}\norm{\mathbb{E}\left[K^{-1}\mathbf{X}_i^T\Tilde{V}^{-1}\mathbf{X}_i\right]\bs_{T} - \mathbb{E}\left[K^{-1}\mathbf{X}_i^T\Tilde{V}^{-1}\mathbf{X}_i\right]\bs}_{\infty}\\
    &= \sup_{T \in N_{\bs_j}} \norm{\Sigma(\bs_{T} - \bs)}_{\infty}\\
    &\leq \norm{\Sigma_{\cdot 1}\boldsymbol{\varepsilon}}_{\infty}\\
    &= \norm{\Sigma_{\cdot 1}}_{\infty}\boldsymbol{\varepsilon}\\
    &\leq 2K_{\Tilde{x}}^2\boldsymbol{\varepsilon} < \infty
\end{align*}
where $K_{\Tilde{x}}$ is defined in Lemma \ref{lemma:knownvlemma8aii}.
\end{remark}
\begin{lemma}[Scaling]
\label{lemma:scaling} Under the conditions of Theorem \ref{theorem:intermediatethm}, this condition holds by plugging in the rates obtained in Conditions \ref{cond:conc} and \ref{cond:l1rates}, adding the following two rate expressions together:
\begin{align*}
    \sqrt{n}O_p(\norm{v}_1\sqrt{\log(p)/n})O_p(s\sqrt{\log(p)/n})&= o_p(1)\\
    \sqrt{n}O_p(\norm{v}_1s_v\sqrt{\log(p)/n})O_p(\sqrt{\log(p)/n})&= o_p(1)
\end{align*}
\end{lemma}
\begin{lemma}[Consistency of $T_n$]
We proceed by showing the following four conditions from \cite{neykov2018unified}:
\begin{enumerate}
\item Condition \ref{cond:conc} from Theorem \ref{theorem:intermediatethm}
\item Condition \ref{cond:l1rates} from Theorem \ref{theorem:intermediatethm}
\item \label{cond:cons2} \begin{equation}
        v^T\left[\mathbb{E}\left[K^{-1}\mathbf{X}_i^T\Tilde{V}^{-1}(\mathbf{X}_i\boldsymbol{\beta}_{T}-\mathbf{Y}_{i})\right]\right]
    \end{equation}
has a unique root at the true parameter value, $\boldsymbol{\beta}_j$.
    \item \label{cond:cons1} \begin{equation}
        \Tilde{v}^T \frac{1}{nK}\sum_{i=1}^n\mathbf{X}_i^T{\Tilde{V}}^{-1}(\mathbf{X}_i\hat{\boldsymbol{\beta}}_T- \mathbf{Y}_{i})
    \end{equation}
            is continuous with a single root or is non-decreasing.

\end{enumerate}
\label{lemma:consistencyconstantv}
\end{lemma}
\begin{proof}
Condition \ref{cond:conc} is shown in Lemma \ref{lemma:assumption1a}, Corollary \ref{corollary:assumption1b}, and Remark \ref{remark:assumption1ce}. Condition \ref{cond:l1rates} is shown in Theorem \ref{theorem:l1rates} and Corollary \ref{corollary:vl1rate}. We next note that 
\[v^T\left[\mathbb{E}\left[K^{-1}\mathbf{X}_i^T{V}^{-1}(\mathbf{X}_i\bs_{T}-\mathbf{Y}_{i})\right]\right] = (T - \boldsymbol{\beta}_j)\]
clearly has a unique root $\boldsymbol{\beta}_j$. Additionally,
\[\Tilde{v}^T \frac{1}{nK}\sum_{i=1}^n\mathbf{X}_i^T{\Tilde{V}}^{-1}(\mathbf{X}_i\hat{\boldsymbol{\beta}}_T- \mathbf{Y}_{i})\]
is linear in $T$ with a unique root except when $\Tilde{v}^T(nK)^{-1}\left[\sum_{i=1}^n\mathbf{X}_i^T{\Tilde{V}}^{-1}\mathbf{X}_i\right]_{\cdot 1} = 0$. But we also know from the definition of $\Tilde{v}$ that 
\[\left|\Tilde{v}^T(nK)^{-1}\left[\sum_{i=1}^n\mathbf{X}_i^T{\Tilde{V}}^{-1}\mathbf{X}_i\right]_{\cdot 1} -1\right| \leq \lambda'\]
so that there will be a unique root as long as $\lambda' < 1$.
\end{proof}

\begin{lemma}
\label{lemma:vstarnormbound}
$\norm{v}_1 \geq (2K_{\Tilde{x}}^2)^{-1}$
\end{lemma}
\begin{proof}
\begin{align*}
    1 &= \norm{v^T \mathbb{E}[K^{-1}\mathbf{X}_i^T\Tilde{V}^{-1}\mathbf{X}_i]_{\cdot 1}}_1\\
    &\leq \norm{v}_1 \norm{\mathbb{E}[K^{-1}\mathbf{X}_i^T\Tilde{V}^{-1}\mathbf{X}_i]_{\cdot 1}}_{\infty}\\
    &\implies \norm{v}_1 \geq (\norm{\mathbb{E}[K^{-1}\mathbf{X}_i^T\Tilde{V}^{-1}\mathbf{X}_i]_{\cdot 1}}_{\infty})^{-1}
\end{align*}
Since each element of $\mathbb{E}[K^{-1}\mathbf{X}_i^T\Tilde{V}^{-1}\mathbf{X}_i]$ is subexponential with subexponential norm $2K_{\Tilde{x}}^2$ (Lemma \ref{lemma:lemma5}) and by definition of the subexponential norm, \\ $2K_{\Tilde{x}}^2 = \norm{\mathbb{E}[K^{-1}\mathbf{X}_i^T\Tilde{V}^{-1}\mathbf{X}_i]_{s,t}}_{\psi_1} \geq |\mathbb{E}[K^{-1}\mathbf{X}_i^T\Tilde{V}^{-1}\mathbf{X}_i]_{s,t}|$ for all $s, t = \{1,...,p\}$, we conclude that
\[\norm{v}_1 \geq (2K_{\Tilde{x}}^2)^{-1}\]
\end{proof}
\begin{lemma}\label{lemma:cltconstantv}[Central Limit Theorem]
Assume that $\Tilde{V}_{ii} \geq C_{\boldsymbol{\varepsilon}} > 0$ for $i = 1...K$, $\mathbf{X}_{ij}\perp \boldsymbol{\varepsilon}_i$, and $\lambda_{\min}(\mathbb{E}[K^{-1}\mathbf{X}_i^T\Tilde{V}^{-1}\mathbf{X}_i]) = \xi > 0$ for a fixed constant $\xi$. Assume also that \[\max(s_v, s)\norm{v}_1\log(p)/\sqrt{n} = o_p(1)\] Then \[n^{1/2}\Delta^{-1/2}\left[v^{T}\sum_{i = 1}^n\mathbf{X}_i^T\Tilde{V}^{-1}\boldsymbol{\varepsilon}_i\right] \overset{d}{\to} N(0,1)\]

\end{lemma}
\begin{proof}
Recall that $\boldsymbol{\varepsilon}_i$ and $\mathbf{X}_{ij}$ are coordinate-wise subgaussian, i.e.:
\[\norm{\boldsymbol{\varepsilon}}_{\psi_2}:= K_{\boldsymbol{\varepsilon}} \text{ and } \sup_{\ell \in 1...p}\norm{X_{ij\ell}}_{\psi_2}:= K_x < \infty\]
for fixed constants $K_{\boldsymbol{\varepsilon}}$ and $K_x$. 

We will show that Lyapunov's condition holds, i.e. 
\begin{equation}\label{eqn:lyapunov2}
    \lim_{n \to \infty} \frac{n^-2}{\Delta^2}\sum_{i = 1}^n \mathbb{E}\left[\left|v^T\frac{1}{K}\mathbf{X}_i{\Tilde{V}}^{-1}(\mathbf{X}_i\boldsymbol{\beta}- \mathbf{Y}_{i})\right|^4\right] \to 0
\end{equation}
Note that $\Delta = v^T\mathbb{E}\left[K^{-1}\mathbf{X}_i^T{\Tilde{V}}^{-1}V{\Tilde{V}}^{-1}\mathbf{X}_i\right]v$ and therefore that \\ $\Delta^2 \geq \norm{v}_2^4\lambda_{\min}(\mathbb{E}\left[K^{-1}\mathbf{X}_i^T{\Tilde{V}}^{-1}V{\Tilde{V}}^{-1}\mathbf{X}_i\right]) = O(1)\norm{v}_2^4$.

\vspace{0.5cm}

Thus, instead of (\ref{eqn:lyapunov2}), we can consider:
\begin{align*}
    \frac{n^-2}{\norm{v}_2^4}\sum_{i = 1}^n \mathbb{E}\left[\left|v^T\frac{1}{K}\mathbf{X}_i{\Tilde{V}}^{-1}(\mathbf{X}_i\boldsymbol{\beta}- \mathbf{Y}_{i})\right|^4\right] &\leq n^{-2}\sum_{i=1}^n\mathbb{E}\left[\norm{K^{-1}(\mathbf{X}_i^T{\Tilde{V}}^{-1}\boldsymbol{\varepsilon}_i)_{s_v}}_2^4\right]\\
    &\leq n^{-1}s_v^2M
\end{align*}
where $M = 2^8(K_{\boldsymbol{\varepsilon}}K_{\Tilde{x}})^4$ and the final inequality follows from Lemma 12 in \cite{neykov2018unified}. Since $\norm{v}_1 \geq (2K_{\Tilde{x}}^2)^{-1}$ (Lemma \ref{lemma:vstarnormbound}), $\max(s_v, s)\norm{v}_1\log(p)/\sqrt{n} = o(1)$ implies $s_v^2/n = o(1)$. To see this, consider:
\begin{align*}
    \max(s_v, s)\norm{v}_1\log(p)/\sqrt{n} = o(1)&\implies  s_v\norm{v}_1\log(p)/\sqrt{n} = o(1)\\
    &\implies s_v2K_{\Tilde{x}}^{-2} \log(p)/\sqrt{n} = o(1)\\
    &\implies s_v\log(p)/\sqrt{n} = o(1)\\
    &\implies s_v/\sqrt{n} = o(1)\\
    &\implies s_v/\sqrt{n} \cdot s_v/\sqrt{n} = s_v^2/n = o(1)
\end{align*}
\end{proof}
\begin{lemma}[Lemma 5]
\label{lemma:lemma5}
With probability at least $1- 2p^{2-\Bar{c}KA_x^2}$,
\[\norm{\frac{1}{n}\sum_{i=1}^n \frac{1}{K}\mathbf{X}_i^{T}{\Tilde{V}}^{-1}\mathbf{X}_i - \mathbb{E}\left[\frac{1}{K} \mathbf{X}_i^{T}{\Tilde{V}}^{-1}\mathbf{X}_i\right]}_{\max} \leq 4A_xK_{\Tilde{x}}^2\sqrt{\log{p}/n}\]
where $\Bar{c} > 0$ is an absolute constant.
\end{lemma}

\begin{proof}
We first note that the $p^2$ elements of $\frac{1}{n}\sum_{i=1}^n \frac{1}{K}\mathbf{X}_i^{T}\Tilde{V}^{-1}\mathbf{X}_i$ are sums of subexponential random variables. This follows from the fact that the rows of $\Tilde{X}_i := \Tilde{V}^{-1/2}\mathbf{X}_i$ remain subgaussian (Lemmas \ref{lemma:subgaussian} and \ref{lemma:scalarsubgaussian}). As stated in \cite{neykov2018unified}, this can be seen through the following relationship from \cite{vershynin2010}:
\[\norm{\Tilde{X}_{ij,\ell}\Tilde{X}_{ij,m}}_{\psi_1} \leq 2\norm{\Tilde{X}_{ij,\ell}}_{\psi_2}\norm{\Tilde{X}_{ij,m}}_{\psi_2} \leq 2K_{\Tilde{x}}^2\]
Therefore, the elements of \[\frac{1}{n}\sum_{i=1}^n \frac{1}{K}\mathbf{X}_i^{T}\Tilde{V}^{-1}\mathbf{X}_i - \mathbb{E}\left[\frac{1}{K} \mathbf{X}_i^{T}\Tilde{V}^{-1}\mathbf{X}_i\right]\] are sums of centered subexponential random variables with subexponential norm:
\[\norm{\Tilde{X}_{ij,\ell}\Tilde{X}_{ij,m} - \mathbb{E}[\Tilde{X}_{ij,\ell}\Tilde{X}_{ij,m}]}_{\psi_1} \leq 4K_{\Tilde{x}}^2\]
The final result follows from a Berstein-type inequality (\cite{vershynin2010}, Proposition 5.16). Although this result depends on the independence of the elements of the sum, we can justify this by summing up the $K$ dependent subexponentials to create $n$ independent subexponentials. Grouping the $K$ dependent terms per observation, we are now summing $n$ centered independent subexponential random variables with subexponential norm less than or equal to $4KK_{\Tilde{x}}^2$. 
\end{proof}
\begin{lemma}\label{lemma:VhatVtildesandwich}$\norm{\hat{V}_n^{-1}V\hat{V}_n^{-1} - \Tilde{V}^{-1}V\Tilde{V}^{-1}}_2 = o_p(1)$
    \begin{proof}
        \begin{align*}
            &\norm{\hat{V}_n^{-1}V( \hat{V}_n^{-1} - \Tilde{V}^{-1}) + \Tilde{V}^{-1}V( \hat{V}_n^{-1} - \Tilde{V}^{-1}) - \Tilde{V}^{-1}V\hat{V}_n^{-1} + \hat{V}_n^{-1}V\Tilde{V}^{-1}}_2\\
            &= \norm{\hat{V}_n^{-1}V( \hat{V}_n^{-1} - \Tilde{V}^{-1}) - \Tilde{V}^{-1}V( \hat{V}_n^{-1} - \Tilde{V}^{-1}) - \Tilde{V}^{-1}V\hat{V}_n^{-1} + \hat{V}_n^{-1}V\Tilde{V}^{-1} + 2\Tilde{V}^{-1}V( \hat{V}_n^{-1} - \Tilde{V}^{-1})}_2\\
            &= \norm{( \hat{V}_n^{-1} - \Tilde{V}^{-1})V( \hat{V}_n^{-1} - \Tilde{V}^{-1}) - \Tilde{V}^{-1}V(\hat{V}_n^{-1} - \Tilde{V}^{-1}) + (\hat{V}_n^{-1} - \Tilde{V}^{-1})V \Tilde{V}^{-1} + 2\Tilde{V}^{-1}V( \hat{V}_n^{-1} - \Tilde{V}^{-1})}_2\\
            &= \norm{( \hat{V}_n^{-1} - \Tilde{V}^{-1})V( \hat{V}_n^{-1} - \Tilde{V}^{-1}) + \Tilde{V}^{-1}V(\hat{V}_n^{-1} - \Tilde{V}^{-1}) + (\hat{V}_n^{-1} - \Tilde{V}^{-1})V \Tilde{V}^{-1} }_2\\
            &\leq \norm{\hat{V}_n^{-1} - \Tilde{V}^{-1}}_2^2\norm{V}_2 + 2\norm{\hat{V}_n^{-1} - \Tilde{V}^{-1}}_2\norm{\Tilde{V}^{-1}}_2\norm{V}_2\\
            &= O_p\left(s\frac{\log(p)}{n}\right)\cdot O(1) + 2 O_p\left(\sqrt{s\frac{\log(p)}{n}}\right) \cdot O(1) \cdot O(1)\\
            &= O_p\left(s\frac{\log(p)}{n}\right) = o_p(1)
        \end{align*}
    \end{proof}
\end{lemma}
\begin{lemma}
\label{lemma:deltaconsistentknownv}
$\hat{\Delta}$ is a consistent estimator of $\Delta$.
\end{lemma}
\begin{proof}

Recall that 
\[\Delta = v^T\mathbb{E}\left[\frac{1}{K} \mathbf{X}_i^T\Tilde{V}^{-1}V{\Tilde{V}}^{-1}\mathbf{X}_i\right]v\] and 
%\[\Tilde{V} = {A}^{1/2}\Bar{R}{A}^{1/2}\]
\[\hat{\Delta} = \frac{1}{n}\hat{v}^{T}\left(\frac{1}{K}\sum_{i=1}^n \mathbf{X}_i^{T}\hat{V}_n^{-1}\left[\frac{1}{n}\sum_{i=1}^n\boldsymbol{\hat{\varepsilon}}_i\boldsymbol{\hat{\varepsilon}}_i^T\right]\hat{V}_n^{-1}\mathbf{X}_i\right)\hat{v} \]
 We consider five different pieces:
\begin{align*}
    %0' &= v^T\mathbb{E}\left[\frac{1}{K} \mathbf{X}_i^{wT}V\mathbf{X}_i^w\right]v - v^T\mathbb{E}\left[\frac{1}{K} \mathbf{X}_i^{T}\Tilde{V}^{-1}V\Tilde{V}^{-1}\mathbf{X}_i\right]v\\
    1' &= \frac{1}{n}v^T\left(\frac{1}{K}\sum_{i=1}^n \mathbf{X}_i^{T}\Tilde{V}^{-1}V{\Tilde{V}}^{-1}\mathbf{X}_i\right)v - v^T\mathbb{E}\left[\frac{1}{K} \mathbf{X}_i^{T}\Tilde{V}^{-1}V\Tilde{V}^{-1}\mathbf{X}_i\right]v \\
    2' &= \frac{1}{n}v^T\left(\frac{1}{K}\sum_{i=1}^n \mathbf{X}_i^{T}\hat{V}_n^{-1}V\hat{V}_n^{-1}\mathbf{X}_i\right)v - \frac{1}{n}v^T\left(\frac{1}{K}\sum_{i=1}^n \mathbf{X}_i^{T}\Tilde{V}^{-1}V\Tilde{V}^{-1}\mathbf{X}_i\right)v\\
    3' &= \frac{1}{n}\hat{v}^{T}\left(\frac{1}{K}\sum_{i=1}^n \mathbf{X}_i^{T}\hat{V}_n^{-1}V\hat{V}_n^{-1}\mathbf{X}_i\right)\hat{v} - \frac{1}{n}v^T\left(\frac{1}{K}\sum_{i=1}^n \mathbf{X}_i^{T}\hat{V}_n^{-1}V\hat{V}_n^{-1}\mathbf{X}_i\right)v\\
    4' &=  \frac{1}{n}\hat{v}^{T}\left(\frac{1}{K}\sum_{i=1}^n \mathbf{X}_i^{T}\hat{V}_n^{-1}\left[\frac{1}{n}\sum_{i=1}^n\boldsymbol{\varepsilon}_i\boldsymbol{\varepsilon}_i^{T}\right]\hat{V}_n^{-1}\mathbf{X}_i\right)\hat{v}  - \frac{1}{n}\hat{v}^{T}\left(\frac{1}{K}\sum_{i=1}^n \mathbf{X}_i^{T}\hat{V}_n^{-1}V\hat{V}_n^{-1}\mathbf{X}_i\right)\hat{v} \\
    5' &=  \frac{1}{n}\hat{v}^{T}\left(\frac{1}{K}\sum_{i=1}^n \mathbf{X}_i^{T}\hat{V}_n^{-1}\left[\frac{1}{n}\sum_{i=1}^n\boldsymbol{\hat{\varepsilon}}_i\boldsymbol{\hat{\varepsilon}}_i^T\right]\hat{V}_n^{-1}\mathbf{X}_i\right)\hat{v}  - \frac{1}{n}\hat{v}^{T}\left(\frac{1}{K}\sum_{i=1}^n \mathbf{X}_i^{T}\hat{V}_n^{-1}\left[\frac{1}{n}\sum_{i=1}^n\boldsymbol{\varepsilon}_i\boldsymbol{\varepsilon}_i^{T}\right]\hat{V}_n^{-1}\mathbf{X}_i\right)\hat{v} \\
\end{align*}
\begin{align*}
    |1'| &= \left|v^T\left[\frac{1}{n}\sum_{i=1}^n \frac{1}{K}\mathbf{X}_i^{T}\Tilde{V}^{-1}V{\Tilde{V}}^{-1}\mathbf{X}_i - \mathbb{E}\left[\frac{1}{K} \mathbf{X}_i^{T}\Tilde{V}^{-1}V{\Tilde{V}}^{-1}\mathbf{X}_i\right]\right]v\right|\\
    &\leq \norm{v}_1^2\norm{\frac{1}{n}\sum_{i=1}^n \frac{1}{K}\mathbf{X}_i^{T}\Tilde{V}^{-1}V{\Tilde{V}}^{-1}\mathbf{X}_i - \mathbb{E}\left[\frac{1}{K} \mathbf{X}_i^{T}\Tilde{V}^{-1}V{\Tilde{V}}^{-1}\mathbf{X}_i\right]}_{\max}\\
    &= \norm{v}_1^2 O_p(\sqrt{\log{p}/n})\\
    &= O(1) \cdot o_p(1) = o_p(1)
\end{align*}
Where the second to last line follows from Lemma \ref{lemma:lemma5} with slight modifications.
% \begin{align*}
%     |0'| &= \left|v^T\left[\mathbb{E}\left[\frac{1}{K} \mathbf{X}_i^{wT}V\mathbf{X}_i^w - \frac{1}{K} \mathbf{X}_i^{T}\Tilde{V}^{-1}V\Tilde{V}^{-1}\mathbf{X}_i\right]\right]v\right|\\
%     &\leq \norm{v}_2^2\norm{\mathbb{E}\left[K^{-1}\mathbf{X}_i^T\hat{V}_n^{-1}V\hat{V}_n^{-1}\mathbf{X}_i - K^{-1}\mathbf{X}_i^T\Tilde{V}^{-1}V\Tilde{V}^{-1}\mathbf{X}_i\right]}_2\\
%     &\leq \norm{v}_2^2\norm{\mathbb{E}\left[K^{-1}\mathbf{X}_i^T\left\{(\hat{V}_n^{-1}-\Tilde{V}^{-1})V(\hat{V}_n^{-1}+\Tilde{V}^{-1}) - \Tilde{V}^{-1}V\hat{V}_n^{-1} + \hat{V}_n^{-1}V\Tilde{V}^{-1}\right\}\mathbf{X}_i\right]}_2\\
%     &\leq \norm{v}_2^2\norm{\mathbb{E}\left[K^{-1}\mathbf{X}_i^T(\hat{V}_n^{-1}-\Tilde{V}^{-1})V(\hat{V}_n^{-1}+\Tilde{V}^{-1})\mathbf{X}_i\right] - \mathbb{E}\left[K^{-1}\mathbf{X}_i^T(\Tilde{V}^{-1}V\hat{V}_n^{-1} - \hat{V}_n^{-1}V\Tilde{V}^{-1})\mathbf{X}_i\right]}_2\\
%   &\leq \norm{v}_2^2
% \end{align*}
\begin{align*}
    |2'| &= \left|\frac{1}{n}v^T\left(\frac{1}{K}\sum_{i=1}^n \mathbf{X}_i^{T}(\hat{V}_n^{-1}V\hat{V}_n^{-1} - \Tilde{V}^{-1}V^{-1}\Tilde{V}^{-1})\mathbf{X}_i\right)v\right|\\
    & \leq \norm{\hat{V}_n^{-1}V\hat{V}_n^{-1} - \Tilde{V}^{-1}V\Tilde{V}^{-1}}_2\norm{v}_1^2 \norm{\frac{1}{nK}\sum_{i=1}^n\mathbf{X}_i^T\mathbf{X}_i}_{\max} \\
        & \leq \norm{\hat{V}_n^{-1}V\hat{V}_n^{-1} - \Tilde{V}^{-1}V\Tilde{V}^{-1}}_2\norm{v}_1^2\left(\norm{\frac{1}{nK}\sum_{i=1}^n\mathbf{X}_i^T\mathbf{X}_i - \mathbb{E}[K^{-1}\mathbf{X}_i^T\mathbf{X}_i]}_{\max} + \norm{\mathbb{E}[K^{-1}\mathbf{X}_i^T\mathbf{X}_i]}_{\max}\right) \\
        &\leq \norm{\hat{V}_n^{-1}V\hat{V}_n^{-1} - \Tilde{V}^{-1}V\Tilde{V}^{-1}}_2\norm{v}_1^2\left(\norm{\frac{1}{nK}\sum_{i=1}^n\mathbf{X}_i^T\mathbf{X}_i - \mathbb{E}[K^{-1}\mathbf{X}_i^T\mathbf{X}_i]}_{\max} + \norm{\Sigma_X}_{2}\right) \\
        &= o_p(1) \cdot O(1) \cdot \left(O_p\left(\sqrt{\frac{\log(p)}{n}}\right)+ O(1)\right) = o_p(1) \\
\end{align*}
Letting $B_n:= \left(\frac{1}{nK}\sum_{i=1}^n \mathbf{X}_i^{T}\hat{V}_n^{-1}V\hat{V}_n^{-1}\mathbf{X}_i\right)$ and $C_n: = \left(\frac{1}{nK}\sum_{i=1}^n \mathbf{X}_i^{T}\Tilde{V}^{-1}V\Tilde{V}^{-1}\mathbf{X}_i\right)$:
\begin{align*}
  |3'|&= \norm{\hat{v}^{T} B_n \hat{v} - v^TB_n\hat{v} + \hat{v}^{T}B_nv - v^TB_nv}_1\\
   &= \norm{(\hat{v} - v)^T(B_n\hat{v}+ B_nv)}_1 \\
   &\leq \norm{\hat{v} - v}_1\norm{B_n(\hat{v}+v)}_{\infty} \\
   &\leq \norm{\hat{v} - v}_1\norm{(B_n-C_n +C_n)(\hat{v}+v)}_{\infty} \\
  &\leq \norm{\hat{v} - v}_1\left(\norm{(B_n-C_n)(\hat{v}+v)}_{\infty} + \norm{C_n(\hat{v}+v)}_{\infty}\right) \\
  &\leq \norm{\hat{v} - v}_1\left(\norm{(B_n-C_n)(\hat{v}-v)}_{\infty} + 2\norm{(B_n-C_n)v}_{\infty}+ \norm{C_n(\hat{v}-v)}_{\infty} + 2\norm{C_nv}_{\infty} \right) \\
  &\leq \norm{\hat{v}-v}^2_{1}\norm{B_n-C_n}_{\max} + 2\norm{v}_{1}\norm{\hat{v}-v}_{1}\norm{B_n-C_n}_{\max}+\\ &\ \norm{\hat{v}-v}^2_1\norm{\Tilde{V}^{-1}V^{-1}\Tilde{V}^{-1}}_2\norm{\frac{1}{nK}\sum_{i=1}^n\mathbf{X}_i^T\mathbf{X}_i}_{\max} +
   2 \norm{\hat{v} - v}_1 \norm{\Tilde{V}^{-1}V^{-1}\Tilde{V}^{-1}}_2 \norm{v}_1\norm{\frac{1}{nK}\sum_{i=1}^n\mathbf{X}_i^T\mathbf{X}_i}_{\max}\\
  &= o_p(1)\left(o_p(1)\norm{B_n-C_n}_{\max} + O(1)\norm{B_n-C_n}_{\max}+ o_p(1) \cdot O(1) \norm{\frac{1}{nK}\sum_{i=1}^n\mathbf{X}_i^T\mathbf{X}_i}_{\max}\right) +\\
  & o_p(1) \cdot O(1) \cdot O(1) \cdot \norm{\frac{1}{nK}\sum_{i=1}^n\mathbf{X}_i^T\mathbf{X}_i}_{\max}\\
   %&\leq \norm{\hat{v} - v}_1 (\norm{B_n\hat{v}}_{\infty}+\norm{B_nv}_{\infty} )\\
   %&\leq O_p\left(\norm{v}_1s_v\sqrt{\log(p)/{n}}\right)((\lambda' + 1) +  \norm{B_nv}_{\infty})\\
   %&\leq O_p\left(\norm{v}_1s_v\sqrt{\log(p)/{n}}\right)(O_p(\norm{v}_1\sqrt{\log(p)/n}) +  \norm{B_nv}_{\infty})\\
\end{align*}

Using Lemma \ref{lemma:Deltanrate}, we can derive that $\norm{B_n-C_n}_{\max} = O_p\left(\frac{\log(p)}{n}\cdot r_2\right)$, where $r_2$ is the rate associated with $\norm{\hat{V}_n^{-1}V^{-1}\hat{V}_n^{-1} - \Tilde{V}^{-1}V^{-1}\Tilde{V}^{-1}}_2$ in Lemma \ref{lemma:VhatVtildesandwich}. Using Lemma \ref{lemma:lemma5}, we know that $\norm{\frac{1}{nK}\sum_{i=1}^n\mathbf{X}_i^T\mathbf{X}_i}_{\max} = \norm{\Sigma_X}_{\max} + O_p\left(\sqrt{\frac{\log(p)}{n}}\right) \leq O(1) + O_p\left(\sqrt{\frac{\log(p)}{n}}\right) $.
Using these results:
\begin{align*}
    &o_p(1)\left(o_p(1)\norm{B_n-C_n}_{\max} + O(1)\norm{B_n-C_n}_{\max}+ o_p(1) \cdot O(1) \norm{\frac{1}{nK}\sum_{i=1}^n\mathbf{X}_i^T\mathbf{X}_i}_{\max}\right) +\\
    & o_p(1) \cdot O(1) \cdot O(1) \cdot \norm{\frac{1}{nK}\sum_{i=1}^n\mathbf{X}_i^T\mathbf{X}_i}_{\max}\\
    &\leq o_p(1)\left(o_p(1)O_p\left(\frac{\log(p)}{n}\cdot r_2\right) + O(1)O_p\left(\frac{\log(p)}{n}\cdot r_2\right)+ o_p(1) \cdot O(1)\left(O(1) + O_p\left(\sqrt{\frac{\log(p)}{n}}\right)\right)\right) +\\
    & o_p(1) \cdot O(1) \cdot O(1) \cdot \left(O(1) + O_p\left(\sqrt{\frac{\log(p)}{n}}\right)\right)\\
\end{align*}
Finally, using our result for $r_2$, we can conclude that $|3'| = o_p(1)$. Letting $\Tilde{a} = \frac{1}{n}\sum_{i=1}^n\boldsymbol{\varepsilon}_i\boldsymbol{\varepsilon}_i^{T} - V$:
\begin{align*}
    |4'| &=  \left|\frac{1}{n}\hat{v}^{T}\left(\frac{1}{K}\sum_{i=1}^n \mathbf{X}_i^{T}\hat{V}_n^{-1}\left[\frac{1}{n}\sum_{i=1}^n\boldsymbol{\varepsilon}_i\boldsymbol{\varepsilon}_i^{T}\right]\hat{V}_n^{-1}\mathbf{X}_i\right)\hat{v}  - \frac{1}{n}\hat{v}^{T}\left(\frac{1}{K}\sum_{i=1}^n \mathbf{X}_i^{T}\hat{V}_n^{-1}V\hat{V}_n^{-1}\mathbf{X}_i\right)\hat{v} \right|\\
    &= \norm{\frac{1}{n}\hat{v}^T\left(\sum_{i = 1}^n\mathbf{X}_i^T\hat{V}_n^{-1}\Tilde{a}\hat{V}_n^{-1}\mathbf{X}_i\right)\hat{v}}_1\\
   &= \norm{\Tilde{a}}_2\norm{\hat{V}_n^{-1}}_2\norm{\frac{1}{n}\hat{v}^T\left(\sum_{i = 1}^n\mathbf{X}_i^T\hat{V}_n^{-1}\mathbf{X}_i\right)\hat{v}}_{\max}\\
   &\leq \norm{\Tilde{a}}_2\norm{\hat{V}_n^{-1}}_2\norm{\hat{v}}_1\norm{\frac{1}{n}\left(\sum_{i = 1}^n\mathbf{X}_i^T\hat{V}_n^{-1}\mathbf{X}_i\right)\hat{v}}_{\max}\\
      &\leq \norm{\Tilde{a}}_2\norm{\hat{V}_n^{-1}}_2\left(\norm{v}_1 +  O_p\left(s_v\frac{\log(p)}{n}\sqrt{\frac{s\log(p)}{n}}\right)\right)\norm{\frac{1}{n}\left(\sum_{i = 1}^n\mathbf{X}_i^T\hat{V}_n^{-1}\mathbf{X}_i\right)\hat{v}}_{\max}\\
      &\leq \norm{\Tilde{a}}_2\norm{\hat{V}_n^{-1}}_2\left(\norm{v}_1 +  O_p\left(s_v\frac{\log(p)}{n}\sqrt{\frac{s\log(p)}{n}}\right)\right)(1 + O_p(\lambda'))\\
      &\leq \norm{\Tilde{a}}_2\left(O(1) + O_p\left(\sqrt{\frac{s\log(p)}{n}}\right)\right)\left(\norm{v}_1 +  O_p\left(s_v\frac{\log(p)}{n}\sqrt{\frac{s\log(p)}{n}}\right)\right)(1 + O_p(\lambda'))\\
      &\leq o_p(1)\left(O(1) + O_p\left(\sqrt{\frac{s\log(p)}{n}}\right)\right)\left(\norm{v}_1 +  O_p\left(s_v\frac{\log(p)}{n}\sqrt{\frac{s\log(p)}{n}}\right)\right)(1 + O_p(\lambda'))\\
    % &\leq \norm{\hat{v}}_1^2\norm{\frac{1}{n}\left(\sum_{i = 1}^n\mathbf{X}_i^T\hat{V}^{-1}\Tilde{a}\hat{V}^{-1}\mathbf{X}_i\right)}_{\max}\\
    % &\leq [\norm{{v}}_1^2+ o_p(1)] \norm{\Tilde{a}}_2\norm{\frac{1}{n}\sum_{i = 1}^n\mathbf{X}_i^T\hat{V}^{-2}\mathbf{X}_i}_{\max}\\
    % &\leq [\norm{{v}}_1^2+ o_p(1)] \norm{\Tilde{a}}_2\norm{\hat{V}^{-2}}_2\frac{1}{n}\max_{(\ell,m)}\norm{\sum_{i = 1}^n \mathbf{X}_i^\ell}_2\norm{\sum_{i = 1}^n \mathbf{X}_i^m}_2\\
    % &\leq [\norm{{v}}_1^2+ o_p(1)] \norm{\Tilde{a}}_2\norm{\hat{V}^{-1}}^2_2K\frac{1}{n}\max_{(\ell,m)}\norm{\sum_{i = 1}^n \mathbf{X}_i^\ell}_{\infty}\norm{\sum_{i = 1}^n \mathbf{X}_i^m}_{\infty}\\
    % &\leq [\norm{{v}}_1^2+ o_p(1)] \norm{\Tilde{a}}_2\norm{\hat{V}^{-1}}^2_2K\frac{1}{n}\max_{\ell = 1...p}\norm{\sum_{i = 1}^n \mathbf{X}_i^\ell}_{\infty}^2\\
    % &= [\norm{{v}}_1^2+ o_p(1)] \norm{\Tilde{a}}_2\left(\norm{\Tilde{V}^{-1}}_2 + o_p(1)\right)^2K\frac{1}{n}\max_{\ell = 1...p}\norm{\sum_{i = 1}^n \mathbf{X}_i^\ell}_{\infty}^2 \rightarrow{\text{Lemma \ref{lemma:workingcor}}}\\
    % &= O_p(1) \cdot \norm{\Tilde{a}}_2 \cdot O_p(1) \cdot O_p(1) \rightarrow{\text{Lemma \ref{lemma:lemma5}}}\\
\end{align*}
Since the dominating term is $o_p(1)$, we can conclude that $|4'| = o_p(1)$. Letting $\Tilde{b} = \frac{1}{n}\sum_{i=1}^n\boldsymbol{\hat{\varepsilon}}_i\boldsymbol{\hat{\varepsilon}}_i^T - \frac{1}{n}\sum_{i=1}^n\boldsymbol{\varepsilon}_i\boldsymbol{\varepsilon}_i^{T}$:
\begin{align*}
    |5'| &= \left|\frac{1}{n}\hat{v}^{T}\left(\frac{1}{K}\sum_{i=1}^n \mathbf{X}_i^{T}\hat{V}_n^{-1}\left[\frac{1}{n}\sum_{i=1}^n\boldsymbol{\hat{\varepsilon}}_i\boldsymbol{\hat{\varepsilon}}_i^T\right]\hat{V}_n^{-1}\mathbf{X}_i\right)\hat{v}  - \frac{1}{n}\hat{v}^{T}\left(\frac{1}{K}\sum_{i=1}^n \mathbf{X}_i^{T}\hat{V}_n^{-1}\left[\frac{1}{n}\sum_{i=1}^n\boldsymbol{\varepsilon}_i\boldsymbol{\varepsilon}_i^{T}\right]\hat{V}_n^{-1}\mathbf{X}_i\right)\hat{v}\right| \\
    &= \norm{\frac{1}{n}\hat{v}^T\left(\sum_{i = 1}^n\mathbf{X}_i^T\hat{V}_n^{-1}\Tilde{b}\hat{V}_n^{-1}\mathbf{X}_i\right)\hat{v}}_1\\
\end{align*}
Clearly showing the rate of this term is parallel to the procedure for $4'$, the only difference being the term $\Tilde{b}$. We therefore focus only on this term:
\begin{align*}
    \Tilde{b} &= \left[\frac{1}{nK}\sum_{i=1}^n(\mathbf{Y}_{i}-\mathbf{X}_i\bh)(\mathbf{Y}_{i}-\mathbf{X}_i\bh)^T\right] - \left[\frac{1}{nK}\sum_{i=1}^n(\mathbf{Y}_{i}-\mathbf{X}_i\bs)(\mathbf{Y}_{i}-\mathbf{X}_i\bs)^T\right]\\
    &= \frac{1}{nK}\sum_{i=1}^n \mathbf{X}_i(\bh \bh^T - \boldsymbol{\beta}\boldsymbol{\beta}^{T})\mathbf{X}_i^T - \mathbf{Y}_{i}(\bh ^T - \boldsymbol{\beta}^{T})\mathbf{X}_i^T - \mathbf{X}_i(\bh - \bs)\mathbf{Y}_{i}^T\\
    &= \frac{1}{nK}\sum_{i=1}^n \mathbf{X}_i(\bh \bh^T - \boldsymbol{\beta}\boldsymbol{\beta}^{T})\mathbf{X}_i^T - \mathbf{Y}_{i}(\bh ^T - \boldsymbol{\beta}^{T})\mathbf{X}_i^T - \mathbf{X}_i(\bh - \bs)\mathbf{Y}_{i}^T + 2\mathbf{X}_i\bs\boldsymbol{\beta}^{T}\mathbf{X}_i^T - 2\mathbf{X}_i\bs\boldsymbol{\beta}^{T}\mathbf{X}_i^T -\\
    &\mathbf{X}_i\bh\boldsymbol{\beta}^{T}\mathbf{X}_i^T + \mathbf{X}_i\bh\boldsymbol{\beta}^{T}\mathbf{X}_i^T - \mathbf{X}_i\bs\bh^T\mathbf{X}_i^T + \mathbf{X}_i\bs\bh^T\mathbf{X}_i^T\\
    &= \frac{1}{nK}\sum_{i=1}^n \underbrace{\left[\mathbf{X}_i(\bh -\bs)\right]\left[\mathbf{X}_i(\bh -\bs)\right]^T}_{(i)} - \underbrace{(\mathbf{Y}_{i}-\mathbf{X}_i\bs)(\bh^T -\boldsymbol{\beta}^{T})\mathbf{X}_i^T}_{(ii)}- \underbrace{\mathbf{X}_i(\bh -\bs)(\mathbf{Y}_{i}-\mathbf{X}_i\bs)^T}_{(iii)}
\end{align*}

We'll show that the $w, j$th element of (i), (ii), and (iii), respectively, converges to zero in probability.
Note that in Lemma \ref{theorem:l1rates}, we showed that 
\[\frac{1}{nK}\norm{X(\bh - \bs)}_2^2 = \frac{1}{nK}\sum_{i = 1}^n\sum_{j=1}^K\mathbf{X}_{ij}^T(\bh-\bs)(\bh-\bs)^T\mathbf{X}_{ij} = O_p(s\log (p)/n) = o_p(1)\]
\begin{align*}
    (i)_{w,j} &= \frac{1}{nK}\sum_{i=1}^n x_{iw}^T(\bh - \bs)(\bh -\bs)^T\mathbf{X}_{ij}\\
     &= \frac{1}{nK}\gamma_w^T\gamma_j\\
     &\leq \sqrt{\frac{1}{nK} \gamma_w^T\gamma_w}\sqrt{\frac{1}{nK} \gamma_j^T\gamma_j} \rightarrow \text{Cauchy-Schwarz}\\
     &= \sqrt{o_p(1)}\sqrt{o_p(1)} = o_p(1) 
\end{align*}
where $x_{iw}$ is the $w$th row of $\mathbf{X}_i$ and $\gamma_w = [x_{1w}^T(\bh - \bs)...x_{nw}^T(\bh - \bs)]^T$. The final line follows from the fact that $\frac{1}{nK}\norm{\gamma_w}_2^2 \leq \frac{1}{nK}\norm{X(\bh - \bs)}_2^2 = o_p(1)$.
\begin{align*}
    (iii)_{w,j} &= \frac{1}{nK}\sum_{i=1}^n x_{iw}^T(\bh - \bs)\boldsymbol{\varepsilon}_{ij}\\
    &\leq \sum_{i=1}^n |x_{iw}^T(\bh - \bs)||\boldsymbol{\varepsilon}_{ij}|\\
    &\leq \norm{\frac{1}{nK}\gamma_w}_2\sqrt{\frac{1}{nK}\sum_{i=1}^n\boldsymbol{\varepsilon}^2_{ij}}\rightarrow \text{Cauchy-Schwarz}\\
    &= o_p(1)O_p(1) = o_p(1) 
\end{align*}
where the final line follows because $\norm{\frac{1}{nK}\gamma_w}_2 \leq \frac{1}{nK}\norm{X(\bh - \bs)}_2 = o_p(1)$ and by the LLN, $\frac{1}{n}\sum_{i=1}^n\boldsymbol{\varepsilon}^2_{ij} = O(1)$.
Adding a finite number of elements, each converging to 0, implies that the Frobenius norm is also converging to 0, i.e. 
\begin{align*}
    \norm{\Tilde{b}}_2 \leq \norm{\Tilde{b}}_F = o_p(1)
\end{align*}
Therefore $|5'| = o_p(1)$.
% where the second to last line follows from Corollary \ref{corollary:vl1rate} and the definition of $\hat{v}$. From Lemma \ref{lemma:knownvlemma8aii}, with probability at least $1 - 2Kp^{2-\Bar{c}A_x^2}$,  $\norm{B_nv}_{\infty} = O_p(\norm{v}_1\sqrt{\log(p)/n})$ and so,
% \begin{align*}
%   |2'|&\leq O_p\left(\norm{v}_1s_v\sqrt{\log(p)/{n}}\right)(O_p(\norm{v}_1\sqrt{\log(p)/n}) +  O_p(\norm{v}_1\sqrt{\log(p)/n}))\\
%   &= O_p\left(\norm{v}_1s_v\sqrt{\log(p)/{n}}\right)(o_p(1) + o_p(1) )\\
%   &= o_p(1)
% \end{align*}

% \begin{align*}
%     3' &=  \frac{1}{n}\hat{v}^{T}\left(\frac{1}{K}\sum_{i=1}^n \mathbf{X}_i^{T}\hat{V}^{-1}\left[\frac{1}{n}\sum_{i=1}^n\boldsymbol{\varepsilon}_i\boldsymbol{\varepsilon}_i^{T}\right]\hat{V}^{-1}\mathbf{X}_i\right)\hat{v}  - \frac{1}{n}\hat{v}^{T}\left(\frac{1}{K}\sum_{i=1}^n \mathbf{X}_i^{T}\hat{V}^{-1}V\hat{V}^{-1}\mathbf{X}_i\right)\hat{v} \\
%     4' &=  \frac{1}{n}\hat{v}^{T}\left(\frac{1}{K}\sum_{i=1}^n \mathbf{X}_i^{T}\hat{V}^{-1}\left[\frac{1}{n}\sum_{i=1}^n\boldsymbol{\hat{\varepsilon}}_i\boldsymbol{\hat{\varepsilon}}_i^T\right]\hat{V}^{-1}\mathbf{X}_i\right)\hat{v}  - \frac{1}{n}\hat{v}^{T}\left(\frac{1}{K}\sum_{i=1}^n \mathbf{X}_i^{T}\hat{V}^{-1}\left[\frac{1}{n}\sum_{i=1}^n\boldsymbol{\varepsilon}_i\boldsymbol{\varepsilon}_i^{T}\right]\hat{V}^{-1}\mathbf{X}_i\right)\hat{v} \\
% \end{align*}
We can thus conclude that:
\begin{align*}
  |\hat{\Delta} - \Delta| &\leq |1'| + |2'| + |3'| +|4'| +|5'| \\
  & = o_p(1) + o_p(1) + o_p(1) +o_p(1) +o_p(1)\\
  &= o_p(1)
\end{align*}
that is, $\hat{\Delta}$ is a consistent estimator of $\Delta$.
\end{proof}

\subsection{Proofs of Lemmas needed for Theorem \ref{theorem:semiparaeffi}}
In this section we show that our estimator, $T_n$, is semiparametrically efficient using extensions of Le Cam's arguments for asymptotically linear
estimators after \cite{jankova2018semiparametric}. 

\begin{lemma}\label{lemma:sp1}The estimator $T_n$ defined in Equation \ref{eqn:theestimate} is asymptotically linear with influence function \[\ell_{\bs}(\mathbf{X}_i) = v^TK^{-1}\mathbf{X}_i^T\Tilde{V}_n^{-1}(\mu_i - \mathbf{Y}_{i})\]
\begin{proof}
Using Lemma 1 from \cite{neykov2018unified}, section E, 
\begin{align*}
    &\hat{v}^T \left[\frac{1}{nK}\sum_{i=1}^n\mathbf{X}_i^T\hat{V}_n^{-1}(\mathbf{X}_i\bh_{T_n} - \mathbf{Y}_{i})\right] =\\  &\hat{v}^T \left[\frac{1}{nK}\sum_{i=1}^n\mathbf{X}_i^T\hat{V}_n^{-1}(\mathbf{X}_i\bh_{\bs_j} - \mathbf{Y}_{i})\right] + \frac{\partial}{\partial T}\left[\hat{v}^T \frac{1}{nK}\sum_{i=1}^n\mathbf{X}_i^T\hat{V}_n^{-1}(\mathbf{X}_i\bh_{\Tilde{T}} - \mathbf{Y}_{i})\right] (T - {\bs_j})\\
    &= v^T\frac{1}{nK}\sum_{i=1}^n \mathbf{X}_i^T\hat{V}_n^{-1}\boldsymbol{\varepsilon}_i +  \frac{\partial}{\partial T} \left[\hat{v}^T \frac{1}{nK}\sum_{i=1}^n\mathbf{X}_i^T\hat{V}_n^{-1}(\mathbf{X}_i\bh_{\Tilde{T}} - \mathbf{Y}_{i})\right] ({T} - {\bs_j}) + o_p(n^{-1/2})\\
    0 &= v^T\frac{1}{nK}\sum_{i=1}^n \mathbf{X}_i^T\hat{V}_n^{-1}\boldsymbol{\varepsilon}_i + \frac{\partial}{\partial T} \left[\hat{v}^T \frac{1}{nK}\sum_{i=1}^n\mathbf{X}_i^T\hat{V}_n^{-1}(\mathbf{X}_i\bh_{\Tilde{T}} - \mathbf{Y}_{i})\right] ({T_n} - {\bs_j}) + o_p(n^{-1/2})\\
\end{align*}
where $\Tilde{T}$ lies on a line connecting $T_n$ and $\bs_j$. We now examine the second term above:

\begin{align*}
    \frac{\partial}{\partial T} \left[\hat{v}^T \frac{1}{nK}\sum_{i=1}^n\mathbf{X}_i^T\hat{V}_n^{-1}(\mathbf{X}_i\bh_{\Tilde{T}} - \mathbf{Y}_{i})\right] ({T_n} - {\bs_j}) &= ({T_n} - {\bs_j}) +\\ &\ \left[\hat{v}^T\frac{1}{nK}\left[\sum_{i = 1}^n \mathbf{X}_i^T\hat{V}_n^{-1}\mathbf{X}_i\right]_{\cdot j} - 1)\right]({T_n} - {\bs_j})\\
    &\leq (T_n - \bs_j) +\\ &\ \norm{\hat{v}^T\frac{1}{nK}\left[\sum_{i = 1}^n \mathbf{X}_i^T\hat{V}_n^{-1}\mathbf{X}_i\right] - e_j}_{\infty}|T_n - {\bs_j}|\\
    &\leq (T_n - \bs_j) + O_p(\sqrt{\log(p)/n})|T_n - \bs_j|\\
    &= (T_n - \bs_j) + O_p(\sqrt{\log(p)/n})O_p(n^{-1/2}) \\
    &= (T_n - \bs_j) + o_p(1)O_p(n^{-1/2})=  (T_n - \bs_j) + o_p(1) 
\end{align*}
Where the rate follows from the scaling assumption (Remark \ref{remark:assumption1ce}) and the root-n consistency of our estimate (Theorem \ref{theorem:cltforunknownv}). Returning to the prior expression (ignoring sign flip):

\begin{align*}
    (T_n - \bs_j) &= v^T\frac{1}{nK}\sum_{i=1}^n \mathbf{X}_i^T\hat{V}_n^{-1}\boldsymbol{\varepsilon}_i + o_p(n^{-1/2})\\
    &= v^T\frac{1}{nK}\sum_{i=1}^n \mathbf{X}_i^T\Tilde{V}^{-1}\boldsymbol{\varepsilon}_i + v^T\frac{1}{nK}\sum_{i=1}^n \mathbf{X}_i^T(\hat{V}_n^{-1}- \Tilde{V}^{-1})\boldsymbol{\varepsilon}_i + o_p(n^{-1/2})\\
    &\leq v^T\frac{1}{nK}\sum_{i=1}^n \mathbf{X}_i^T\Tilde{V}^{-1}\boldsymbol{\varepsilon}_i + \norm{\hat{V}_n^{-1}- \Tilde{V}^{-1}}_2\norm{v}_1\norm{\frac{1}{nK}\sum_{i=1}^n \mathbf{X}_i^T\boldsymbol{\varepsilon}_i}_{\infty} + o_p(n^{-1/2})\\
    &= O_p\left(\sqrt{\frac{s\log(p)}{n}}\right) \cdot O(1) \cdot O_p\left(\sqrt{\frac{\log(p)}{n}}\right) \\
    &= o_p(n^{-1/2})
\end{align*}
where the last line follows from our scaling assumptions.
Thus, our estimator is asymptotically linear with influence function

\[\ell_{\bs}(\mathbf{X}_i) = v^TK^{-1}\mathbf{X}_i^T\Tilde{V}^{-1}(\mu_i - \mathbf{Y}_{i})\]
\end{proof}
\end{lemma}
\begin{lemma}[Bounding the variance]
\label{lemma:sp2}
Assume that $\norm{v}_1 < \infty$. It follows that $\mathbb{E}[\ell^2_{\bs}(\mathbf{X}_i)] = Var(\ell_{\bs}) < \infty$.
\end{lemma}
\begin{proof}
    \begin{align*}
    Var(\ell_{\bs}) &= \mathbb{E}\left[v^TK^{-1}\mathbf{X}_i^T\Tilde{V}^{-1}V\Tilde{V}^{-1}\mathbf{X}_iv\right]\\
    &= v^TK^{-2}\mathbb{E}\left[\mathbf{X}_i^T\Tilde{V}^{-1}V\Tilde{V}^{-1}\mathbf{X}_i\right]v\\
    &\leq \norm{v}_1^2K^{-2}\norm{\mathbb{E}\left[\mathbf{X}_i^T\Tilde{V}^{-1}V\Tilde{V}^{-1}\mathbf{X}_i\right]}_{\max}\\
    &\leq \norm{v}_1^2 \norm{\Tilde{V}^{-1}V\Tilde{V}^{-1}}_2\norm{\Sigma_X}_2\\
    &\leq O(1) \cdot O(1) \cdot O(1) < \infty
\end{align*}

\end{proof}

\begin{lemma}\label{lemma:sp3}
The score function associated with the data-generating process is differentiable and bounded.
\begin{proof}
Assuming Gaussianity, the score for a single observation is:
\begin{align*}
    s_{\boldsymbol{\beta}} &= \mathbf{X}_i^TV^{-1}\boldsymbol{\varepsilon}_i\\
    \ddot{s_{\boldsymbol{\beta}}} &= 0
\end{align*}
So the score is twice differentiable and the second derivative uniformly bounded by any positive constant. 
\end{proof}
\end{lemma}
\begin{lemma}\label{lemma:sp4}
The information matrix for Gaussian data is:

\[I_{\bs} = \mathbb{E}[s_{\bs}s_{\bs}^T] = \mathbb{E}[\mathbf{X}_i^TV^{-1}\mathbf{X}_i]\]

The eigenvalues of this matrix are bounded from assumptions in Lemma \ref{lemma:cltconstantv}. Also,
\begin{align*}
    \norm{\frac{1}{n}\sum_{i=1}^n \mathbf{X}_i^TV^{-1}\mathbf{X}_i - \mathbb{E}[\mathbf{X}_i^TV^{-1}\mathbf{X}_i]}_{\max} &= o_p(1) 
\end{align*}
by small modifications to the arguments of Lemma \ref{lemma:lemma5}.
\end{lemma}
\begin{lemma}\label{lemma:sp5}
$\forall \boldsymbol{\varepsilon} >0$
\[\lim_{n \to \infty} \mathbb{E}_{\bs}[f^2_{\bs}\mathbb{I}_{|f_{\bs}|>\boldsymbol{\varepsilon}\sqrt{n}}] = 0\]
where $f_{\bs}(x):= \ell_{\bs}(x) + h^Ts_{\bs}(x)$ for $x \in \mathbb{R}^p$ and $h = \sqrt{n}(\Tilde{\boldsymbol{\beta}} - \bs)$ for $\Tilde{\boldsymbol{\beta}}\in B(\bs, \frac{c}{\sqrt{n}})$. Here $B(\bs, \frac{c}{\sqrt{n}})$ is an $L_2$ ball of length $\frac{c}{\sqrt{n}}$ around $\bs$.
\begin{proof}
\begin{align*}
 \mathbb{E}_{\bs}[f^2_{\bs}\mathbb{I}_{|f_{\bs}|>\boldsymbol{\varepsilon}\sqrt{n}}] &= \mathbb{E}_{\bs}[f^2_{\bs}]\mathbb{P}({|f_{\bs}|>\boldsymbol{\varepsilon}\sqrt{n}})\\
\mathbb{P}({|f_{\bs}|>\boldsymbol{\varepsilon}\sqrt{n}})&= \mathbb{P}({|(v^T+\sqrt{n}(\Tilde{\boldsymbol{\beta}} - \bs))(\mathbf{X}_i^TV^{-1}\boldsymbol{\varepsilon}_i)|>\boldsymbol{\varepsilon}\sqrt{n}})
\end{align*}
Since $v^T+\sqrt{n}(\Tilde{\boldsymbol{\beta}} - \bs) = O_p(1)$ and $\mathbf{X}_i^TV^{-1}\boldsymbol{\varepsilon}_i$ doesn't change with $n$, we conclude:
\[\lim_{n \to \infty} \mathbb{P}({|f_{\bs}|>\boldsymbol{\varepsilon}\sqrt{n}}) = 0\]
\begin{align*}
    \mathbb{E}_{\bs}[f^2_{\bs}(x)] &= \mathbb{E}_{\bs}(v^T+\sqrt{n}(\Tilde{\boldsymbol{\beta}} - \bs))(\mathbf{X}_i^TV^{-1}\boldsymbol{\varepsilon}_i)(\mathbf{X}_i^TV^{-1}\boldsymbol{\varepsilon}_i)^T(v^T+\sqrt{n}(\Tilde{\boldsymbol{\beta}} - \bs))^T]\\
    &= (v^T+\sqrt{n}(\Tilde{\boldsymbol{\beta}} - \bs))\mathbb{E}_{\bs}[\mathbf{X}_i^TV^{-1}\mathbf{X}_i](v^T+\sqrt{n}(\Tilde{\boldsymbol{\beta}} - \bs))^T
\end{align*}
Since the eigenvalues of $\mathbb{E}_{\bs}[\mathbf{X}_i^TV^{-1}\mathbf{X}_i]$ are bounded and $(v^T+\sqrt{n}(\Tilde{\boldsymbol{\beta}} - \bs)) = O_p(1)$, we can conclude that $\forall \boldsymbol{\varepsilon} >0$
\[\lim_{n \to \infty} \mathbb{E}_{\bs}[f^2_{\bs}\mathbb{I}_{|f_{\bs}|>\boldsymbol{\varepsilon}\sqrt{n}}] = 0\]

\end{proof}
\end{lemma}
\begin{lemma}\label{lemma:sp6}
    $\forall h \in \mathbb{R}^p$:
\[\mathbb{E}_{\bs}[\ell_{\bs}h^Ts_{\bs}] - h^Te_j = o(1)\]
\begin{proof}
\begin{align*}
    \mathbb{E}_{\bs}[\ell_{\bs}h^Ts_{\bs}] &= \mathbb{E}_{\bs}[v^T\mathbf{X}_i^TV^{-1}\boldsymbol{\varepsilon}_ih^T\mathbf{X}_i^TV^{-1}\boldsymbol{\varepsilon}_i]\\
    &= v^T\mathbb{E}_{\bs}[\mathbf{X}_i^TV^{-1}\mathbb{E}[\boldsymbol{\varepsilon}_i \boldsymbol{\varepsilon}_i^{T}|\mathbb{X}]V^{-1}\mathbf{X}_ih]\\
    &= v^T\mathbb{E}_{\bs}[\mathbf{X}_i^TV^{-1}\mathbf{X}_i]h
\end{align*}
Since $v = \left[\mathbb{E}_{\bs}[\mathbf{X}_i^TV^{-1}\mathbf{X}_i]\right]^{-1}_{\cdot j}$:
\[v^T\mathbb{E}_{\bs}[\mathbf{X}_i^TV^{-1}\mathbf{X}_i]h = h^Te_j\] 
as was desired to show.
\end{proof}
\end{lemma}

\begin{lemma}\label{lemma:sp7}
$\boldsymbol{\beta}+ I^{-1}_{\bs}e_j/\sqrt{n} \in B(\bs, c/\sqrt{n})$. 
\begin{proof}
This is equivalent to showing:

\begin{align*}
    \norm{\mathbb{E}_{\bs}[\mathbf{X}_i^TV^{-1}\mathbf{X}_i]^{-1}_{\cdot j}}_2 \leq c
\end{align*}
which follows directly from the fact that $\lambda^{-1}_{\min}(\Sigma_X) = O(1)$ and $\lambda^{-1}_{\min}(V) = O(1)$.
\end{proof}
\end{lemma}
\end{document}